\documentclass[12pt]{article}%
\usepackage{graphicx}%
\usepackage[normalem]{ulem}
\usepackage{amsmath,amsthm,amsfonts,
amsbsy,amssymb,upref,enumerate,bigstrut,
color,mathtools,mathrsfs,float,bm,dsfont,scalefnt,mathtools}
\usepackage{lipsum}
\usepackage{lmodern}
\usepackage{natbib}
\usepackage{stmaryrd}
\usepackage{authblk} 
\usepackage{subfigure}
\usepackage{array}
\usepackage{multirow}
\usepackage{booktabs}
\usepackage[colorlinks=true,linkcolor=blue,citecolor=blue]{hyperref}

\usepackage[left=1in,top=1.2in,right=0.7in]{geometry}

\theoremstyle{definition}

\newtheorem{theorem}{Theorem}[section]

\newtheorem{corollary}[theorem]{Corollary}

\newtheorem{lemma}[theorem]{Lemma}
\newtheorem{proposition}[theorem]{Proposition}
\newtheorem{remark}[theorem]{Remark}
\numberwithin{equation}{section} 

\makeatletter
\def\@seccntformat#1{\@ifundefined{#1@cntformat}%
	{\csname the#1\endcsname\quad}
	{\csname #1@cntformat\endcsname}
}
\makeatother

\newif\ifShowComments
\ShowCommentstrue
\def\strutdepth{\dp\strutbox}
\def\druk#1{\strut\vadjust{\kern-\strutdepth
        {\vtop to \strutdepth{%
                \baselineskip\strutdepth\vss
                        \llap{\hbox{#1}\quad}\null}}}}

\title{\bf
Bivariate log-symmetric models: distributional properties, parameter estimation and an application to fatigue data analysis
}

\author[1, 2]{Roberto Vila \thanks{rovig161@gmail.com}} 
\author[2]{\, Narayanaswamy Balakrishnan \thanks{bala@mcmaster.ca
} }
\author[1]{\, Helton Saulo  \thanks{heltonsaulo@gmail.com} }
\author[1]{Ana Protazio \thanks{ana.protazio@gmail.com}}
\affil[1]{Department of Statistics, University of
	 Bras\'ilia, Bras\'ilia, Brazil}
\affil[2]{
Department of Mathematics and Statistics, McMaster University, Hamilton, Ontario, Canada}
\begin{document}
\maketitle
\newlength\mywidth
\begin{abstract}
{
The bivariate Gaussian distribution has been a key model for many developments in statistics. However, many real-world phenomena generate data that follow asymmetric distributions, and consequently bivariate normal model is inappropriate in such situations. Bidimensional log-symmetric models have attractive properties and can be considered as good alternatives in these cases. In this paper, we discuss bivariate log-symmetric distributions and their characterizations. We establish several distributional properties and obtain the maximum likelihood estimators of the model parameters. A Monte Carlo simulation study is performed for examining the performance of the developed parameter estimation method. A real data set is finally analyzed to illustrate the proposed model and the associated inferential method.
}

\end{abstract}


\section{Introduction}
\noindent

A distribution is said to be log-symmetric when the corresponding random variable and its reciprocal have the same distribution \cite[see][]{Jones2008}. A characterization of distributions of this type can be constructed by taking the exponential function of a symmetric random variable. Hence, log-symmetric distributions may be used to describe strictly positive data. The class of log-symmetric distributions is quite broad and includes a large portion of bimodal distributions and those with lighter or heavier tails than the log-normal distribution; see \citet{Vanegas2016}. Some examples of log-symmetric distributions are log-normal, log-Student-$t$, log-logistic, log-Laplace, log-power-exponential, log-slash, etc; see \citet{Crow1988}, \cite{Jones2008}, and \cite{Vanegas2016}, for pertinent details.

Another important feature of the log-symmetric class is that they are closed under scale change and under reciprocity, according to  \cite{Puig2008}, which are also desirable properties for distributions that are used to describe strictly positive data. Further, log-symmetric models allow one to model the median or the asymmetry (relative dispersion).

In addition, the log-symmetric class has many other desirable statistical properties. For example, the two parameters of the log-symmetric distribution are orthogonal and they can be interpreted directly as median and skewness (or relative dispersion, taking into account two parameters that are interpreted as measures of position and scale, as stated by \cite{Vanegas2016}, which are, in the context of asymmetric distributions, the ones that mean the most, being complete measures of location and shape, respectively.

The main objective of this work is to extend in a natural way the definition of univariate log-symmetric distributions to the bivariate case, to
study their main statistical properties, to develop maximum likelihood (ML) method for the estimation of model parameters, and finally to show its applicability to the analysis of fatigue data.

The rest of this work is organized as follows. In Section \ref{Sec:2}, the bivariate log-symmetric (BLS) model is proposed. In Section \ref{Sec:3}, some mathematical properties such as stochastic representation, quantile function, conditional distribution, Mahalanobis distance, independence, moments, correlation function, among others, are discussed.
In Section \ref{Sec:4}, we describe the ML method for the estimation of the BLS model parameters.
In Section \ref{Sec:5}, we perform a Monte Carlo simulation study for evaluating the performance of the ML estimators.
In Section \ref{Sec:6}, we apply the BLS models to a data set on material fatigue to demonstrate the practical utility of the BLS models. Finally, in Section \ref{Sec:7}, some concluding remarks are made.

\section{Bivariate log-symmetric model}\label{Sec:2}
\noindent

A continuous random vector $\boldsymbol{T}=(T_1,T_2)$ is said to follow a bivariate log-symmetric (BLS) distribution if its joint probability density function (PDF) is given by
\begin{eqnarray}\label{PDF}
	f_{T_1,T_2}(t_1,t_2;\boldsymbol{\theta})
	=
	{1\over t_1t_2\sigma_1\sigma_2\sqrt{1-\rho^2}Z_{g_c}}\,
	g_c\Biggl(
	{\widetilde{t_1}^2-2\rho\widetilde{t_1}\widetilde{t_2}+\widetilde{t_2}^2
	\over 
	1-\rho^2}
	\Biggr),
	\quad 
	t_1,t_2>0,
\end{eqnarray}
where
$$
	\widetilde{t_i}
	=
	\log\biggl[\Bigl({t_i\over \eta_i}\Bigr)^{1/\sigma_i}\biggr], \ \eta_i=\exp(\mu_i), \ i=1,2, \nonumber
$$
with $\boldsymbol{\theta}=(\eta_1,\eta_2,\sigma_1,\sigma_2,\rho)$ being the parameter vector, $\mu_i\in\mathbb{R}$, $\sigma_i>0$, $i=1,2$; $\rho\in(-1,1)$
; $Z_{g_c}>0$ is the partition function given by
\begin{align}\label{partition function}
Z_{g_c}
&=
\int_{0}^{\infty}\int_{0}^{\infty}
	{1\over t_1t_2\sigma_1\sigma_2\sqrt{1-\rho^2}}\,
	g_c\Biggl(
{\widetilde{t_1}^2-2\rho\widetilde{t_1}\widetilde{t_2}+\widetilde{t_2}^2
	\over 
	1-\rho^2}
\Biggr)\, {\rm d}t_1{\rm d}t_2,
\end{align}
and $g_c$ is a scalar function, referred to as the density generator; see  \cite{Fang1990}.
We use, in this case, the notation $\boldsymbol{T}\sim {\rm BLS}(\boldsymbol{\theta},g_c)$. 
In this paper, we prove that, when it exists, the variance-covariance matrix of a random vector $\boldsymbol{T}\sim {\rm BLS}(\boldsymbol{\theta},g_c)$, denoted by $K_{\boldsymbol{T}}$, is a matrix function of the following dispersion matrix
	(see Subsections \ref{moments} and \ref{Correlation function}):
\begin{align*}
\boldsymbol{\Sigma}=
	\begin{pmatrix}
	\sigma_1^2 & \rho\sigma_1\sigma_2
	\\
	\rho\sigma_1\sigma_2 & \sigma_2^2
	\end{pmatrix}.
\end{align*}
In other words, $K_{\boldsymbol{T}}=\psi(\boldsymbol{\Sigma})$ for some matrix function $\psi: \mathcal{M}_{2,2}\longmapsto \mathcal{M}_{2,2}$, where $\mathcal{M}_{2,2}$ denotes the set of all 2-by-2 real matrices.

Based on the works  \cite{Saulo2017} and  \cite{Vanegas2016}, Table \ref{table:1} presents some examples of bivariate log-symmetric distributions, and some BLS PDF plots are presented in Figure \ref{fig:bls-pdfs}.

\begin{table}[H]
	\caption{Partition functions $(Z_{g_c})$ and density generators $(g_c)$ for some BLS distributions.}
	\vspace*{0.15cm}
	\centering 
	\begin{tabular}{llll} 
		\hline
		Distribution 
		& $Z_{g_c}$ & $g_c$ & Parameter 
		\\ [0.5ex] 
		\noalign{\hrule height 1.7pt}
		Bivariate Log-normal
		& $2\pi$ & $\exp(-x/2)$ & $-$ 
		\\ [1ex] 
		Bivariate Log-Student-$t$
		& ${{\Gamma({\nu/ 2})}\nu\pi\over{\Gamma({(\nu+2)/ 2})}}$  
		& $(1+{x\over\nu})^{-(\nu+2)/ 2}$  &  $\nu>0$
		\\ [1ex]
		Bivariate Log-Pearson Type VII 
		& ${\Gamma(\xi-1)\theta\pi\over\Gamma(\xi)}$  & $(1+{x\over\theta})^{-\xi}$  & $\xi>1$, $\theta>0$
		\\ [1ex]
		Bivariate Log-hyperbolic
		& ${2\pi (\nu+1)\exp(-\nu)\over \nu^2}$ & $\exp(-\nu\sqrt{1+x})$ &  $\nu>0$
		\\ [1ex]   
		Bivariate Log-Laplace 
		& $\pi$  & $K_0(\sqrt{2x})$ & $-$
		\\ [1ex]   
		Bivariate Log-slash
		& ${\pi\over \nu-1}\, 2^{3-\nu\over 2}$ & $ x^{-{\nu+1\over 2}} \gamma({\nu+1\over 2},{x\over 2})$ & $\nu>1$
		\\ [1ex]   
		Bivariate Log-power-exponential
		& ${2^{\xi+1}(1+\xi)\Gamma(1+\xi)}\pi$  & ${\exp\bigl(-{1\over 2}\, x^{1/(1+\xi)}\bigr)}$ & $-1<\xi\leqslant 1$
		\\ [1ex]   
		{Bivariate Log-Logistic}
		& $\pi/2$  & {${\exp(-x)\over (1+\exp(-x))^2}$} &  {$-$}
	\\ [1ex] 
	\hline	
	\end{tabular}
	\label{table:1} 
\end{table}
\noindent
In Table \ref{table:1} above, $\Gamma(t)=\int_0^\infty x^{t-1} \exp(-x) \,{\rm d}x$, $t>0$, is the complete gamma function,
$K_0(u)=\int_0^\infty t^{-1} \exp(-t-{u^2\over 4t}) \,{\rm d}t/2$, $u>0$, is the Bessel function of the third kind
\cite[for more details on some properties of $K_0$, see appendix of ][]{Kotz2001},
and $\gamma(s,x)=\int_{0}^{x}t^{s-1}\exp(-t)\,{\rm d}t$ is the lower incomplete gamma function. 
%

\begin{figure}[htb!]
	\centering
	\subfigure[log-normal]{\includegraphics[scale=0.35]{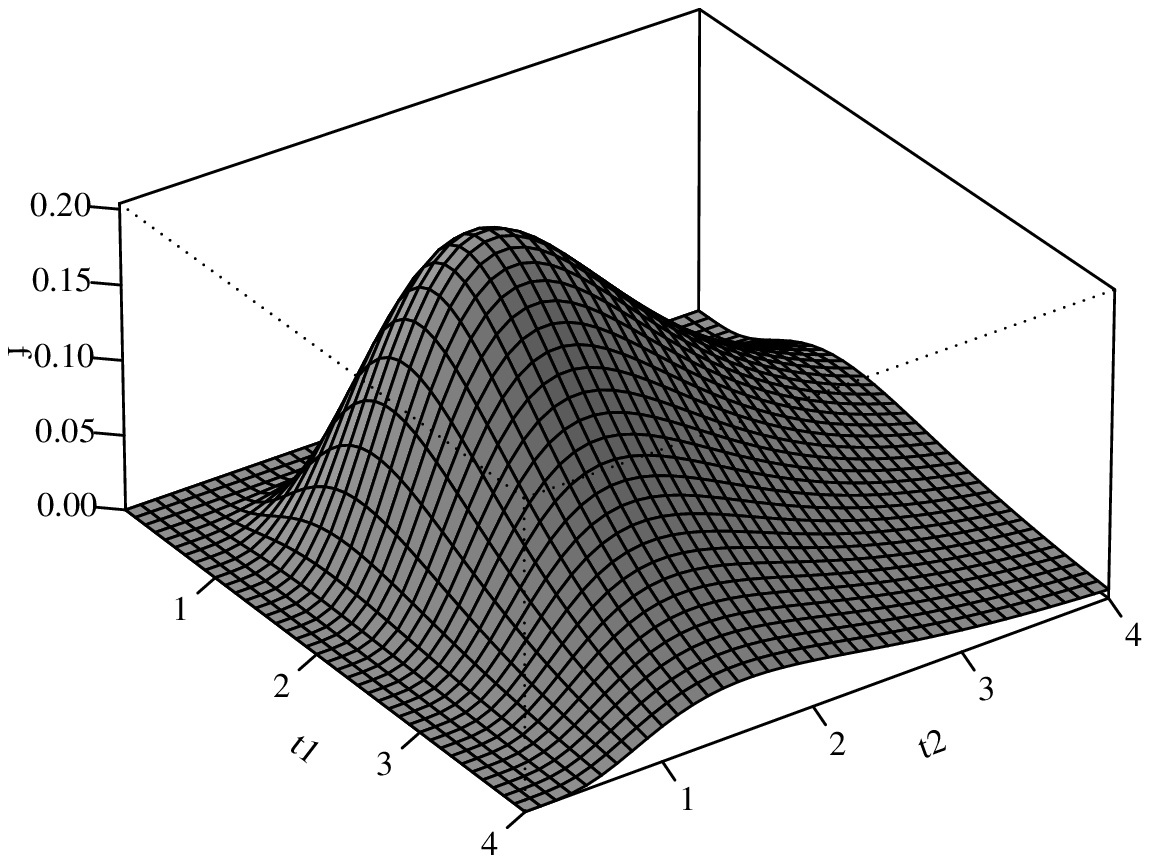}}
	\subfigure[log-Student-$t$ ($\nu = 3$)]{\includegraphics[scale=0.35]{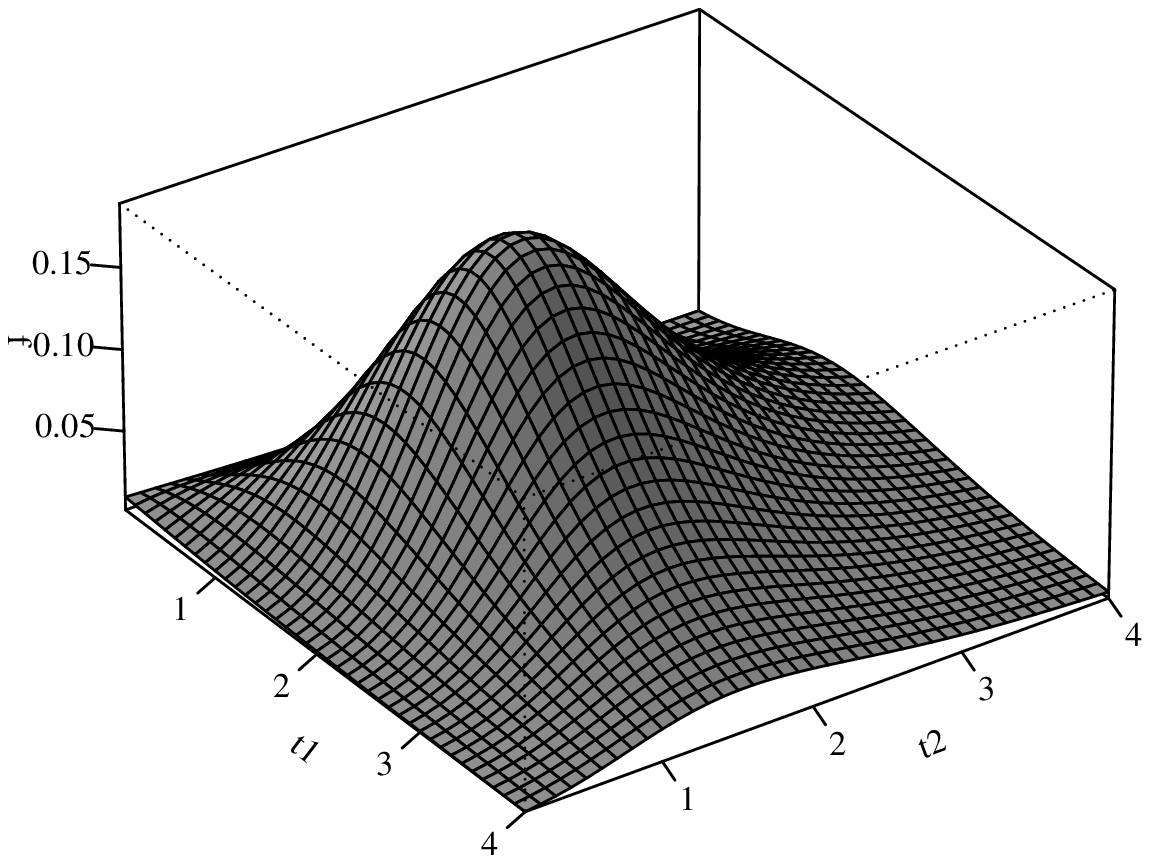}}
	\subfigure[Log-Pearson Type VII ($\xi = 5,\theta = 22$)]{\includegraphics[scale=0.35]{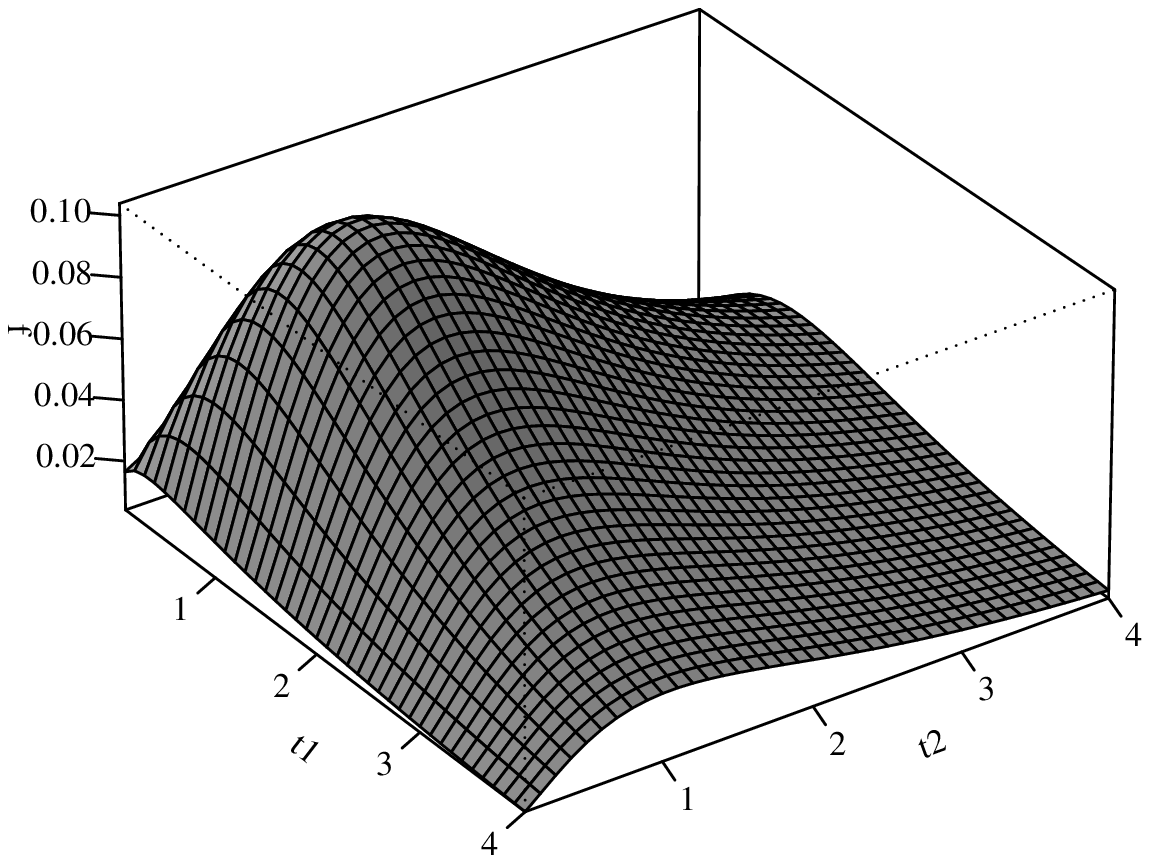}}
	\subfigure[log-hyperbolic ($\nu = 2$)]{\includegraphics[scale=0.35]{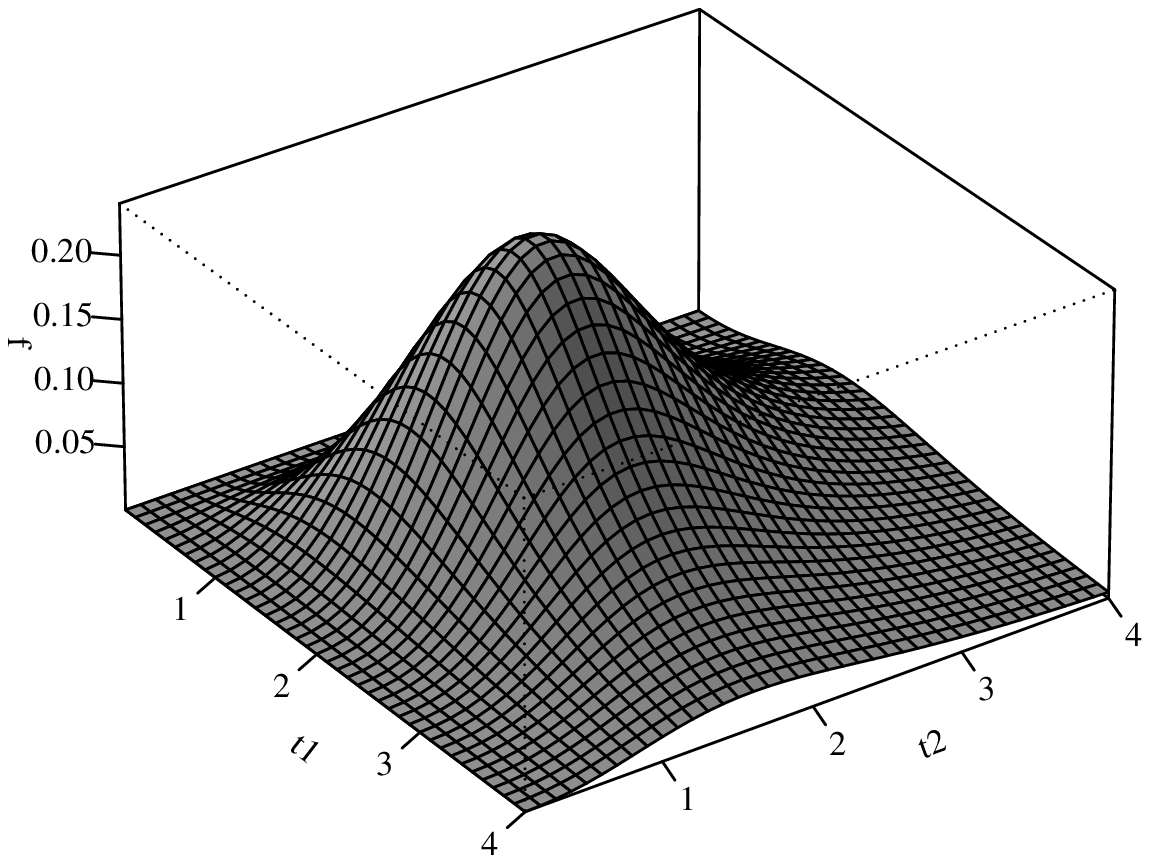}}
	\subfigure[log-slash ($\nu = 4$)]{\includegraphics[scale=0.35]{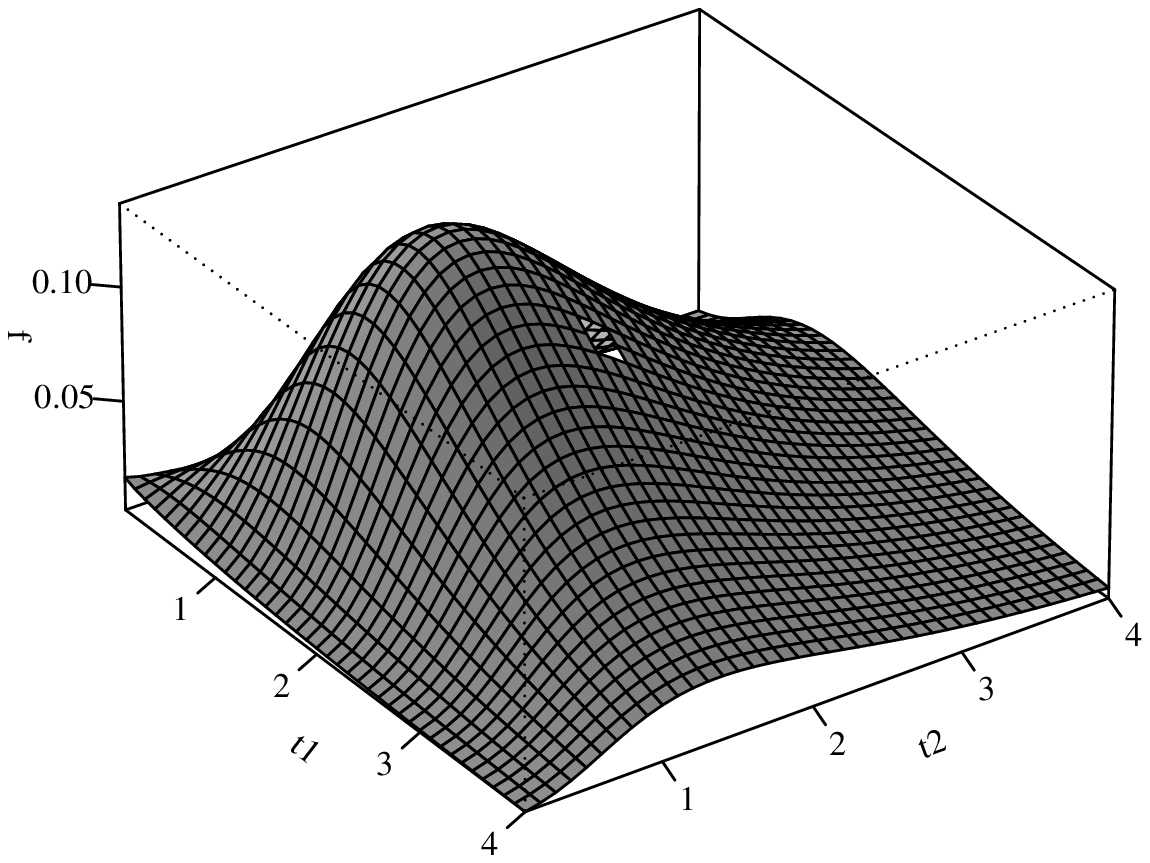}}
	\subfigure[log-power-exponential ($\xi = 0.5$)]{\includegraphics[scale=0.35]{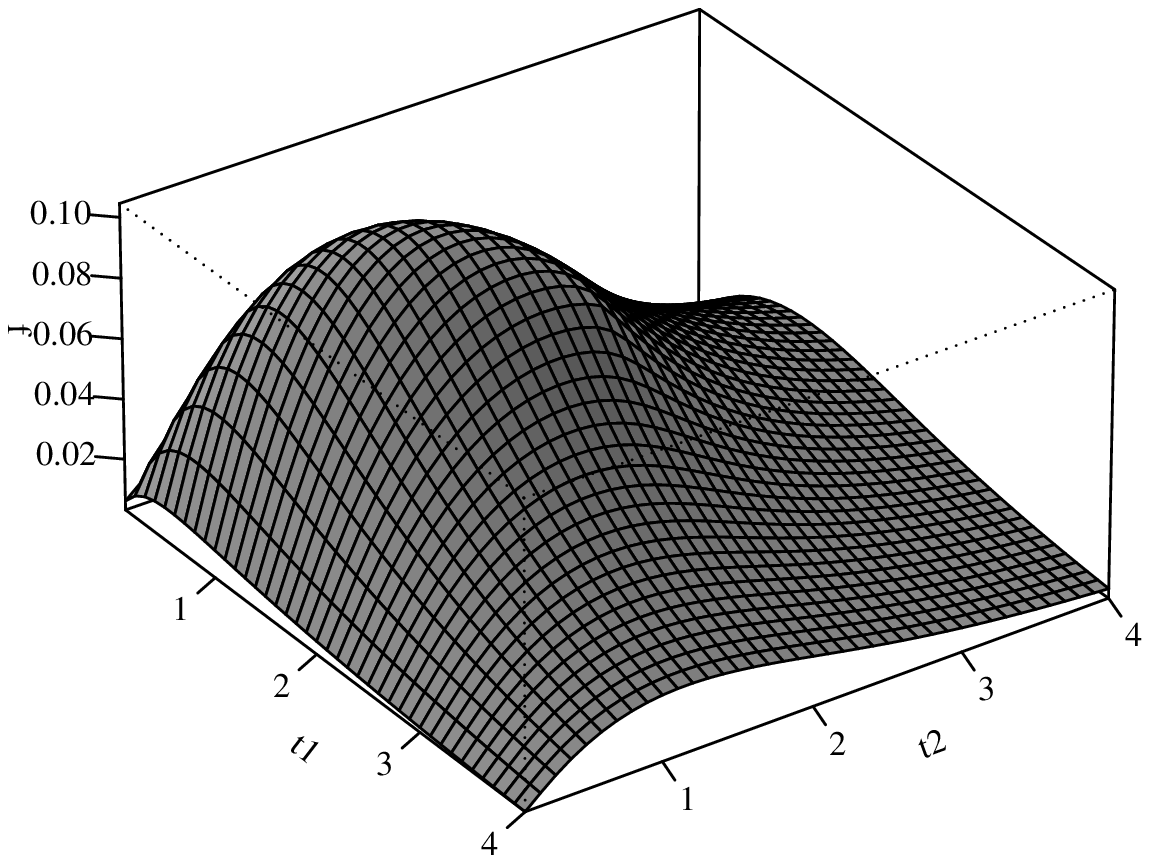}}
	\subfigure[log-logistic]{\includegraphics[scale=0.35]{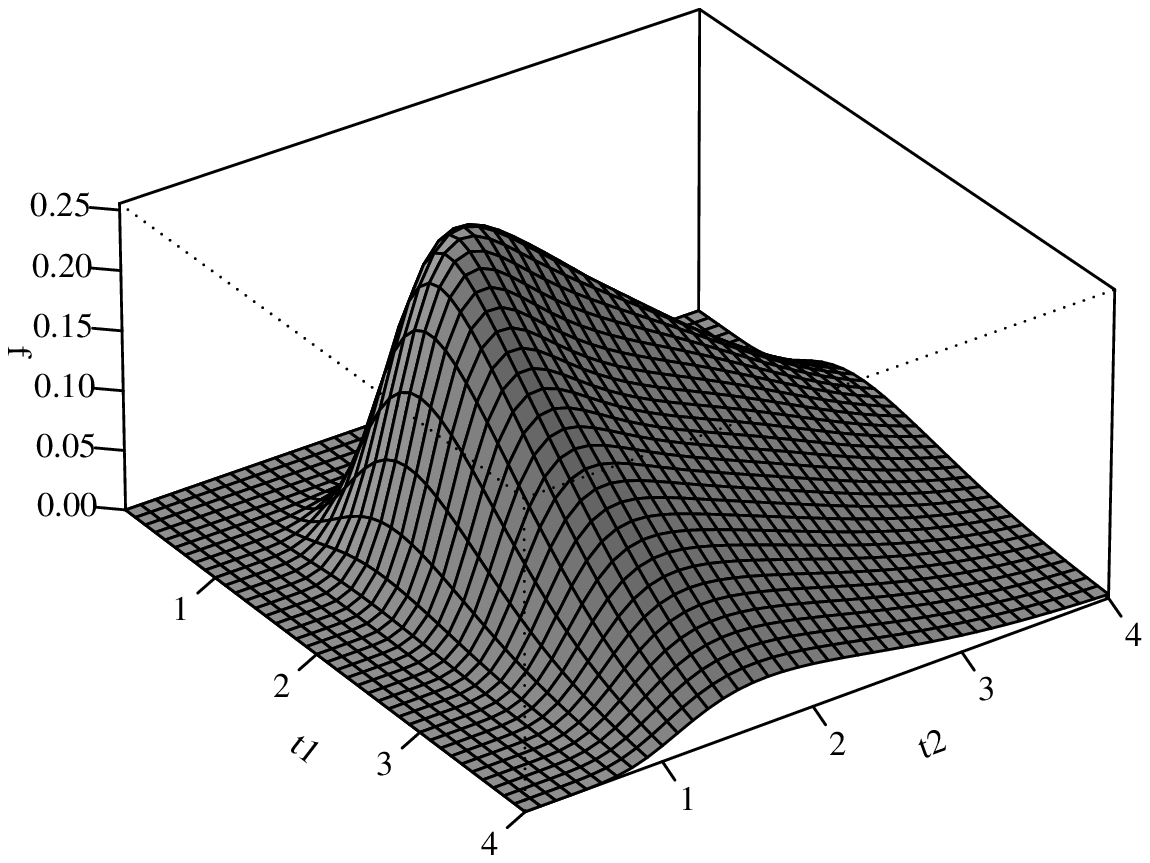}}
	\subfigure[log-Laplace]{\includegraphics[scale=0.35]{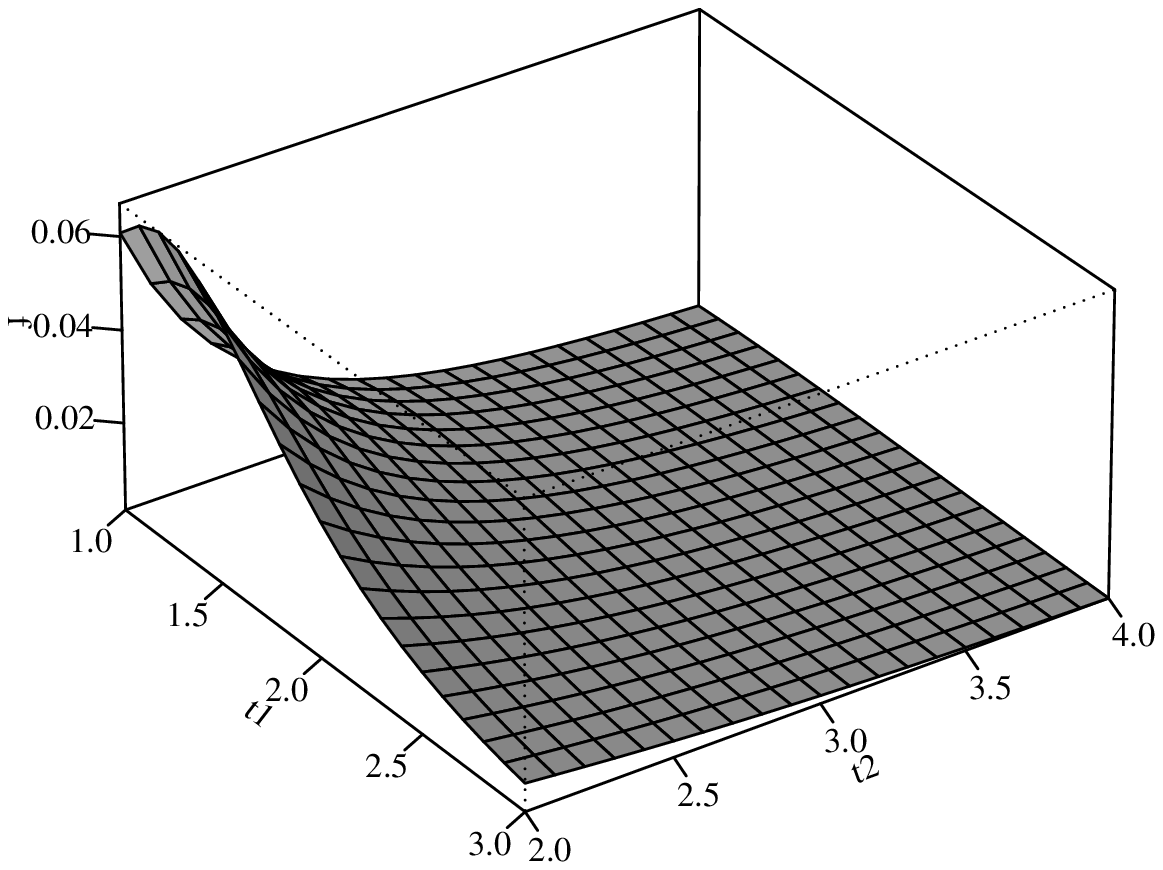}}
	\caption{BLS PDFs with $\boldsymbol{\theta}_*=(2,2,0.5,0.5,0)$.}
	\label{fig:bls-pdfs}
\end{figure}

From \eqref{PDF}, it is clear that the random vector $\boldsymbol{X}=(X_1,X_2)$, with $X_i=\log(T_i)$, $i=1,2$, has a bivariate elliptically symmetric (BES) distribution; see p. 592 of   \cite{Bala2009}. In other words, the PDF of $\boldsymbol{X}$ is
\begin{eqnarray}\label{PDF-symmetric}
f_{X_1,X_2}(x_1,x_2;\boldsymbol{\theta}_*)
=
{1\over \sigma_1\sigma_2\sqrt{1-\rho^2}Z_{g_c}}\,
g_c\Biggl(
{\widetilde{x_1}^2-2\rho\widetilde{x_1}\widetilde{x_2}+\widetilde{x_2}^2
	\over 
	1-\rho^2}
\Biggr),
\quad 
-\infty<x_1,x_2<\infty,
\end{eqnarray}
where
$$
\widetilde{x_i}={x_i-\mu_i\over\sigma_i}, \ i=1,2, \nonumber
$$
with $\boldsymbol{\theta}_*=(\mu_1,\mu_2,\sigma_1,\sigma_2,\rho)$ being the parameter vector and $Z_{g_c}$ is the partition function given in \eqref{partition function}. In this case, we shall use the notation $\boldsymbol{X}\sim {\rm BES}(\boldsymbol{\theta}_*, g_c)$.

A simple calculation shows that the joint cumulative distribution function (CDF) of $\boldsymbol{T}\sim {\rm BLS}(\boldsymbol{\theta},g_c)$, denoted by $F_{T_1,T_2}(t_1,t_2;\boldsymbol{\theta})$, can be expressed as
\begin{align*}
F_{T_1,T_2}(t_1,t_2;\boldsymbol{\theta})
=
F_{X_1,X_2}\big(\log(t_1),\log(t_2);\boldsymbol{\theta}_*\big),
\end{align*}
with $F_{X_1,X_2}(x_1,x_2;\boldsymbol{\theta}_*)$ being the CDF of $\boldsymbol{X}\sim {\rm BES}(\boldsymbol{\theta}_*, g_c)$. Except in the case of bivariate normal, there is no closed-form expression for the CDF of $\boldsymbol{X}$.

\section{Some basic properties} \label{Sec:3}
\noindent

In this section, some distributional properties of the proposed bivariate log-symmetric distributions are discussed.

\subsection{Characterization of the partition function $Z_{g_c}$}

\begin{proposition}\label{partition function-simplicado}
	The partition function $Z_{g_c}$ given in \eqref{partition function} is independent of the parameter vector $\boldsymbol{\theta}$. To be specific, we have
	\begin{align*}
	Z_{g_c}
	=
	\int_{-\infty}^{\infty}\int_{-\infty}^{\infty} 
	g_c\big({z_1}^2+{z_2}^2\big)
	\, {\rm d}{z_1}{\rm d}{z_2}
	=
	\pi 
	\int_{0}^{\infty}
	g_c(u)
	\, {\rm d}{u}.
	\end{align*}
\end{proposition}
\begin{proof}
	The proof of the first identity follows by considering in \eqref{partition function} the following change of variables through the Jacobian method:
	\begin{align*}
	z_1=\widetilde{t}_1,
	\quad 
	z_2={\widetilde{t}_2-\rho \widetilde{t}_1\over\sqrt{1-\rho^2}},
	\end{align*} 
	where $\widetilde{t}_i$, $i=1,2$, are as defined in \eqref{PDF}.
	
	The proof of the second identity follows by using
	integration in polar coordinates: $z_1=r\cos(\theta)$, $z_2=r\sin(\theta)$, with $r\geqslant  0$ and $0\leqslant\theta\leqslant 2\pi$, and then the change of variables $u=r^2$, ${\rm d}u=2r {\rm d}r$.  As this is easy to establish, we omit its proof here.
\end{proof}

%
%

\subsection{Stochastic representation}

\begin{proposition}\label{Stochastic Representation}
The random vector $\boldsymbol{T}=(T_1,T_2)$ has a BLS distribution if
\begin{align*}
&T_1=\eta_1 \exp(\sigma_1 Z_1),
\\[0,2cm]
&T_2=\eta_2 
\exp\big(\sigma_2 {\rho} Z_1+\sigma_2\sqrt{1-\rho^2} Z_2\big),
\end{align*} 
where $Z_1=RDU_1$ and $Z_2=R\sqrt{1-D^2}U_2$;
$U_1$, $U_2$,$R$, and $D$ are mutually independent random variables, $\rho\in(-1,1)$, $\eta_i=\exp(\mu_i)$, and $\mathbb{P}(U_i = -1) = \mathbb{P}(U_i = 1) = 1/2$, $i=1,2$. The random variable $D$ is positive and has its PDF as
\begin{align*}
f_D(d)={2\over \pi\sqrt{1-d^2}}, \quad d\in(0,1).
\end{align*}
Moreover, the positive random variable $R$ is called the generator of the elliptical random vector $\boldsymbol{X}=(X_1,X_2)$.
{In other words, $R$ has is PDF as
\begin{align*}
	f_R(r)={2r g_c(r^2)\over \int_{0}^{\infty}
		g_c(u)
		\, {\rm d}{u}}, \quad r>0.
\end{align*}
}
\end{proposition}
\begin{proof}
It is well-known that \cite[see][]{Abdous2005} the vector $\boldsymbol{X}$ has a BES distribution if
\begin{align}\label{rep-stoch-biv-gaussian}
\begin{array}{llll}
&X_1=\mu_1+\sigma_1 Z_1,
\\[0,2cm]
&X_2=\mu_2+\sigma_2 {\rho} Z_1+\sigma_2\sqrt{1-\rho^2} Z_2.
\end{array}
\end{align}
As $X_i=\log(T_i)$, $i=1,2$, the proof readily follows.
\end{proof}

\subsection{Marginal quantiles}

Let $\boldsymbol{T}=(T_1,T_2)\sim {\rm BLS}(\boldsymbol{\theta},g_c)$ and $p\in(0,1)$.
By using the stochastic representation of Proposition \ref{Stochastic Representation}, we readily obtain
\begin{align*}
p=\mathbb{P}(T_1\leqslant Q_{T_1})
=
\mathbb{P}\big(\eta_1 \exp(\sigma_1 Z_1)\leqslant Q_{T_1}\big)
=
\mathbb{P}\left(Z_1\leqslant \log\biggl[\Bigl({Q_{T_1}\over \eta_1}\Bigl)^{1/\sigma_1}\biggl]\right)
\end{align*}
and 
\begin{align*}
p=\mathbb{P}(T_2\leqslant Q_{T_2})
&=
\mathbb{P}\big(
\eta_2 
\exp(\sigma_2 {\rho} Z_1+\sigma_2\sqrt{1-\rho^2} Z_2)
\leqslant Q_{T_2}\big) 
\\[0,2cm]
&=
\mathbb{P}\left(
{\rho} Z_1+\sqrt{1-\rho^2} Z_2
\leqslant \log\biggl[\Bigl({Q_{T_2}\over \eta_2}\Bigl)^{1/\sigma_2}\biggl]\right).
%
\end{align*}
Hence, the $p$-th quantile of $T_1$ and the $p$-th quantile of $T_2$ are given by
\begin{align*}
\log\biggl[\Bigl({Q_{T_1}\over \eta_1}\Bigl)^{1/\sigma_1}\biggl]
=
Q_{Z_1} 
\quad \Longleftrightarrow \quad 
Q_{T_1}=\eta_1 \exp(\sigma_1 Q_{Z_1})
\end{align*}
and
\begin{align*}
	\log\biggl[\Bigl({Q_{T_2}\over \eta_2}\Bigl)^{1/\sigma_2}\biggl]
	=
	Q_{{\rho} Z_1+\sqrt{1-\rho^2} Z_2}
	\quad \Longleftrightarrow \quad 
	Q_{T_2}=
	\eta_2 \exp(\sigma_2 Q_{{\rho} Z_1+\sqrt{1-\rho^2} Z_2}),
\end{align*}
respectively.

\subsection{Conditional distribution}

\begin{proposition}\label{joint and marginal pdfs}
	The joint PDF of $Z_1$ and $Z_2$, given  in Proposition \ref{Stochastic Representation}, is
	\begin{align*}
	f_{Z_1,Z_2}(z_1,z_2)={g_c(z_1^2+z_2^2)\over \pi \int_{0}^{\infty}
		g_c(u)
		\, {\rm d}{u}}, \quad -\infty<z_1,z_2<\infty.
	\end{align*}
	Moreover, the marginal PDFs of $Z_1$ and $Z_2$, denoted by $f_{Z{_1}}$ and $f_{Z{_2}}$, respectively, are
	\begin{align*}
	f_{Z{_1}}(z_1)=
	{{\displaystyle\int_{\vert z_1\vert}^{\infty}} {2 g_c(w^2)\over \sqrt{1-{z_1^2\over w^2}}}\, {\rm d}{w} 
		\over \pi \int_{0}^{\infty}
		g_c(u)
		\, {\rm d}{u}}
	\quad \text{and} \quad
		f_{Z{_2}}(z_2)=
	{{\displaystyle
		\int_{\vert z_2\vert }^{\infty}} {2 g_c(w^2)\over \sqrt{1-{z_2^2\over w^2}}}\, {\rm d}{w}
		\over \pi \int_{0}^{\infty}
		g_c(u)
		\, {\rm d}{u}}, \quad -\infty<z_1,z_2<\infty.
	\end{align*}
	
Then, $f_{Z_i}(z_i\,\vert\, Z_j=z_j)$ ($i\neq j$) and $f_{Z_i}(z_i)$, $i,j=1,2$, are both even functions.
\end{proposition}
\begin{proof}
For more details on the proof, see Propositions 3.1 and 3.2 of reference \cite{Saulo2022}.
\end{proof}

Let $X$ and $Y$ be two continuous random variables with joint PDF $f_{X,Y}$, and marginal PDFs $f_X$ and $f_Y$, respectively. 
\begin{itemize}
\item
Let $B$ be a Borel subset of $\mathbb{R}$.
The conditional CDF of $X$, given $\{Y\in B\}$, denoted by ${F}_X(x\,\vert\, Y\in B)$,  is defined as (for every $x$)
\begin{align}\label{cdf-cond}
{F}_X(x\,\vert\, Y\in B)=\mathbb{P}(X\leqslant x\,\vert\, Y\in B)
=
\int_{-\infty}^x {f}_X(u\,\vert\, Y\in B)\, {\rm d} u \quad \text{if} \ \mathbb{P}(Y\in B)>0,
%
\end{align}
where ${f}_X(u\,\vert\, Y\in B)$ is the corresponding conditional PDF given by
\begin{align*}
{f}_X(u\,\vert\, Y\in B)=\frac{\int_B {f}_{X,Y}(u,v)\, {\rm d} v}{\mathbb{P}(Y\in B)}.
\end{align*}
We use $X\,\vert\, Y\in B$ to denote the random variable $X$ having the conditional CDF in \eqref{cdf-cond}, given $\{Y\in B\}$;

\item
Let $\varepsilon>0$, and further $\mathbb{P}(y-\varepsilon<Y\leqslant y+\varepsilon)>0$. Then, we define the conditional CDF of $X$, given $Y=y$, denoted by ${F}_X(x\,\vert\, Y=y)$,  as (for every $x$)
\begin{align*}
{F}_X(x\vert Y=y)
=
\lim_{\varepsilon\to 0^+} \mathbb{P}(X\leqslant x\,\vert\, y-\varepsilon<Y\leqslant y+\varepsilon),
\end{align*}
provided that the limit exists.
If the limit exists, there is a nonnegative function ${f}_X(u\vert Y=y)$ (called the conditional PDF) so that (for every $x$)
\begin{align*}
	{F}_X(x\,\vert\, Y=y)
	=
	\int_{-\infty}^x {f}_X(u\,\vert\, Y=y)\, {\rm d} u.
\end{align*}
At every point $(x,y)$ at which $f_{X,Y}$ is continuous and $f_Y(y)>0$ is continuous, the PDF ${f}_X(u\,\vert\, Y=y)$ exists and is given by (see Theorem 6, p. 109, of \cite{Rohatgi2015})
\begin{align*}
{f}_X(u\,\vert\, Y=y)={f_{X,Y}(u,y)\over f_Y(y)}.
\end{align*}
For simplicity, we shall denote it by $X\,\vert\, Y=y$.
\end{itemize}

\begin{lemma}\label{conditional PDF}
	If $\boldsymbol{T}=(T_1,T_2)\sim {\rm BLS}(\boldsymbol{\theta},g_c)$, then
the PDF of $T_2 \vert T_1=t_1$ is given by
\begin{align}\label{cond-pdf}
f_{T_2}(t_2\,\vert\, T_1=t_1)
=
{1\over t_2\sigma_2 \sqrt{1-\rho^2}}\,
{
	f_{Z_2}\biggl(
	{1\over\sqrt{1-\rho^2}}\,\widetilde{t}_2
	-
	{\rho\over\sqrt{1-\rho^2}}\, \widetilde{t}_1 \, \bigg\vert\, Z_1=\widetilde{t}_1 
	\biggr)
},
\end{align}
where $\widetilde{t}_i$, $i=1,2$, are as defined in \eqref{PDF}, and  ${Z_1}$ and ${Z_2}$  are as in Proposition \ref{Stochastic Representation}.
\end{lemma}
\begin{proof}
If $T_1=t_1$, then $Z_1 = \log\big[(t_1/\eta_1)^{1/\sigma_1}\big]=\widetilde{t}_1$. Thus, the conditional distribution of $T_2$, given $T_1=t_1$, is the same as the distribution of
\begin{align*}
	\eta_2 
	\exp\big(\sigma_2 {\rho} \widetilde{t}_1+\sigma_2\sqrt{1-\rho^2} Z_2\big)\, \big\vert\, T_1=t_1.
\end{align*}
Consequently, 
\begin{align*}
F_{T_2}(t_2\, \vert\, T_1=t_1)
&=
\mathbb{P}\big(\eta_2 
\exp(\sigma_2 {\rho} \widetilde{t}_1+\sigma_2\sqrt{1-\rho^2} Z_2)\leqslant t_2\, \big\vert\, T_1=t_1\big)
\\[0,2cm]
&=
\mathbb{P}\biggl(
Z_2\leqslant {1\over\sqrt{1-\rho^2}}\, \widetilde{t}_2-{{\rho}\over\sqrt{1-\rho^2}}\, \widetilde{t}_1\, \bigg\vert\, Z_1= \widetilde{t}_1\biggr).
\end{align*}
Then, the conditional PDF of $T_2$, given  $T_1=t_1$, follows readily.
\end{proof}

\begin{theorem}\label{theo-pdf-cond}
For a Borel subset $B$ of $(0,\infty)$, we define the following Borel set:
\begin{align}\label{B-child}
{B}_\rho
=
{1\over\sqrt{1-\rho^2}}\,
\log\biggl[\Bigl({B\over \eta_2}\Bigr)^{1/\sigma_2}\biggr]
-
{\rho\over\sqrt{1-\rho^2}}\, \widetilde{t_1},
\end{align}
where $\widetilde{t}_1$ is as in \eqref{PDF}.
If $\boldsymbol{T}=(T_1,T_2)\sim {\rm BLS}(\boldsymbol{\theta},g_c)$, then the PDF of $T_1\,\vert\, T_2\in B$ is given by
\begin{align*}
		f_{T_1}(t_1\vert T_2\in B)
		=
{1\over t_1 \sigma_1}\,f_{Z_1}(\widetilde{t}_1)\, 
{
	\int_{{B}_\rho}
	f_{Z_2}(
	w \, \vert\, Z_1=\widetilde{t}_1 )\, 
	{\rm d}w
	\over 
	\mathbb{P}(\rho Z_1+\sqrt{1-\rho^2}Z_2\in {B}_0)
},
\end{align*}
with ${Z_1}$ and ${Z_2}$ as in Proposition \ref{Stochastic Representation}.
\end{theorem}
\begin{proof}
Let $B$ be a Borel subset of $(0,\infty)$. Note that
\begin{align*}
	f_{T_1}(t_1\,\vert\, T_2\in B)
	=
	f_{T_1}(t_1)\, {\int_B f_{T_2}(t_2\,\vert\, T_1=t_1)\, {\rm d}t_2\over \mathbb{P}(T_2\in B)}.
\end{align*}
As $f_{T_1}(t_1)=f_{Z_1}(\widetilde{t}_1)/(\sigma_1 t_1)$ and $\mathbb{P}(T_2\in B)=\mathbb{P}(\rho Z_1+\sqrt{1-\rho^2}Z_2\in {B}_0)$, where ${B}_0$ is as given in \eqref{B-child} with $\rho=0$, the term on the right-hand side of the above identity equals
\begin{align*}
	{1\over \sigma_1 t_1}\,f_{Z_1}(\widetilde{t}_1)\, {\int_B f_{T_2}(t_2\,\vert\, T_1=t_1)\, {\rm d}t_2\over \mathbb{P}(\rho Z_1+\sqrt{1-\rho^2}Z_2\in {B}_0)}.
\end{align*}
By using the expression for $f_{T_2}(t_2\,\vert\, T_1=t_1)$ provided in Lemma \ref{conditional PDF}, the above expression equals
\begin{align*}
{1\over t_1 \sigma_1\sigma_2 \sqrt{1-\rho^2}}\,f_{Z_1}(\widetilde{t}_1)\, 
{\int_B 
	{1\over t_2}\,	
	f_{Z_2}\Big(
	{1\over\sqrt{1-\rho^2}}\,\widetilde{t}_2
	-
	{\rho\over\sqrt{1-\rho^2}}\, \widetilde{t}_1 \, \Big\vert\, Z_1=\widetilde{t}_1 
	\Big)\, 
	{\rm d}t_2\over \mathbb{P}(\rho Z_1+\sqrt{1-\rho^2}Z_2\in {B}_0)},
\end{align*}
where $\widetilde{t}_i$, $i=1,2$, are as in \eqref{PDF}.
Finally, by applying the change of variable $w=(\widetilde{t}_2
-
{\rho}\, \widetilde{t}_1)/\sqrt{1-\rho^2}$, the above expression reduces to
\begin{align*}
{1 \over t_1 \sigma_1}\,f_{Z_1}(\widetilde{t}_1)\, 
{\int_{{B}_\rho}
	f_{Z_2}(
	w \, \vert\, Z_1=\widetilde{t}_1 )\, 
	{\rm d}w
	\over \mathbb{P}(\rho Z_1+\sqrt{1-\rho^2}Z_2\in {B}_0)},
\end{align*}
in which ${B}_\rho$ is as in \eqref{B-child}.
This completes the proof of the theorem.
\end{proof}

\begin{corollary}[Gaussian generator] \label{Gaussian generator}
Let $\boldsymbol{T}=(T_1,T_2)\sim {\rm BLS}(\boldsymbol{\theta},g_c)$ and $g_c(x)=\exp(-x/2)$ be the generator of the bivariate log-normal distribution. Then, for each Borel subset $B$ of $(0,\infty)$, the PDF of $T_1\,\vert\, T_2\in B$ is given by (for $t_1>0$)
\begin{align*}
f_{T_1}(t_1\vert T_2\in B)
=
{1\over t_1 \sigma_1}\,
\phi\Bigl(
\log\Bigl[\Big({t_1\over \eta_1}\Big)^{1/\sigma_1}\Bigr]
\Bigr)\, 
\dfrac{
	\Phi\Bigl(
{1\over\sqrt{1-\rho^2}}
\Big\{
\log\Bigl[\big({B\over \eta_2}\big)^{1/\sigma_2}\Bigr]
-
\rho \log\Big[\big({t_1\over \eta_1}\big)^{1/\sigma_1}\Big]
\Big\}
\Bigr)
}{
	\Phi\Big(\log\Big[\big(\frac{B}{ \eta_2}\big)^{1/\sigma_2}\Big]\Big)
},
\end{align*}
where we use the notation $\Phi(C)=\int_C \phi(x){\rm d}x$, with $\phi(x)=g_c(x^2)/\sqrt{2\pi}$.
\end{corollary}
\begin{proof}
It is well-known that the bivariate normal distribution admits a stochastic representation as in \eqref{rep-stoch-biv-gaussian},  where $Z_1\sim N(0,1)$ and $Z_2\sim N(0,1)$ are independent. Consequently, $Z_2\vert Z_1=z\sim N(0,1)$ and $\rho Z_1+\sqrt{1-\rho^2}Z_2\sim N(0,1)$.
Further, a simple algebraic manipulation shows that
\begin{align*}
	\int_{{B}_\rho}
	f_{Z_2}(
	w \, \vert\, Z_1=\widetilde{t}_1 )\, 
	{\rm d}w
	= 
	\Phi(B_\rho),
\end{align*}
where $B_\rho$ is the Borel set defined in \eqref{B-child}.
Then, by using Theorem \ref{theo-pdf-cond}, the required result follows.
\end{proof}

\begin{corollary}[Student-$t$ generator]
	Let $\boldsymbol{T}=(T_1,T_2)\sim {\rm BLS}(\boldsymbol{\theta},g_c)$ and $g_c(x)=(1+(x/\nu))^{-(\nu+2)/2}$, $\nu>0$, be the generator of the bivariate log-Student-$t$ distribution with $\nu$ degrees of freedom. Then, for each Borel subset $B$ of $(0,\infty)$, the PDF of $T_1\vert T_2\in B$ is given by (for $t_1>0$)
	\begin{align*}
\resizebox{17.5cm}{!}{
	$
	f_{T_1}(t_1\vert T_2\in B)
	=
	{1 \over  t_1 \sigma_1}\,
	f_{\nu}\Bigl(
	\log\Bigl[\big({t_1\over \eta_1}\big)^{1/\sigma_1}\Bigr]
	\Bigr)\, 
	\dfrac{
F_{\nu+1}\Big(\sqrt{\nu+1\over \nu+\widetilde{t_1}^2}\, {1\over\sqrt{1-\rho^2}}
\Big\{
\log\Bigl[\big({B\over \eta_2}\big)^{1/\sigma_2}\Bigr]
-
\rho \log\Big[\big({t_1\over \eta_1}\big)^{1/\sigma_1}\Big]
\Big\} 
\Big)		
}{
		F_\nu\Big(\log\Big[\big(\frac{B}{ \eta_2}\big)^{1/\sigma_2}\Big]\Big)
	},
$
}
	\end{align*}
	where  $F_\nu(C)=\int_C f_\nu(x){\rm d}x$, with $f_\nu(x)$ being the standard Student-$t$ PDF.
\end{corollary}
\begin{proof}
It is well-known that the bivariate Student-$t$ distribution has a stochastic representation as in \eqref{rep-stoch-biv-gaussian}, where $Z_1=Z_1^* \sqrt{\nu/Q}\sim t_\nu$ and $Z_2= Z_2^* \sqrt{\nu/Q}\sim t_\nu$, $Q\sim\chi^2_\nu$ (chi-square with $\nu$ degrees of freedom) is independent of $Z_1^*$ and ${\rho} Z_1^* +\sqrt{1-\rho^2} Z_2^*$; here, $Z_1^*$ and $Z_2^*$ are independent and identically distributed  standard normal  random variables.

As $\rho  Z_1^*+\sqrt{1-\rho^2} Z_2^*\sim N(0,1)$, we have
$\rho Z_1+\sqrt{1-\rho^2}Z_2=(\rho  Z_1^*+\sqrt{1-\rho^2} Z_2^*)\sqrt{\nu/Q}\sim t_\nu$. Then,
\begin{align*}
	\mathbb{P}\big(\rho Z_1+\sqrt{1-\rho^2}Z_2\in {B}_0\big)=F_\nu({B}_0).
\end{align*}

On the other hand, if $\boldsymbol{X}=(X_1,X_2)\sim {\rm BES}(\boldsymbol{\theta}_*, g_c)$, by Remark 3.7 of  \cite{Saulo2022}, we have
\begin{align*}
	\sqrt{{\nu+1\over (\nu+r^2)(1-\rho^2)}} \biggl({X_2-\mu_2\over \sigma_2}-\rho r\biggr)\, \Bigg\vert \, {X_1-\mu_1\over \sigma_1}=r\sim t_{\nu+1}.
\end{align*}
Consequently,
\begin{align*}
\mathbb{P}(T_{\nu+1}\leqslant x)
&=
	\mathbb{P}\left(\sqrt{{\nu+1\over (\nu+r^2)(1-\rho^2)}} \biggl({X_2-\mu_2\over \sigma_2}-\rho r\biggr)\leqslant x\, \Bigg\vert \, {X_1-\mu_1\over \sigma_1}=r\right)
	\\[0,2cm]
	&=
	\mathbb{P}\left(Z_2 \leqslant \sqrt{{\nu+r^2\over \nu+1}}\, x\, \Bigg\vert \, Z_1=r\right), \quad T_{\nu+1}\sim t_{\nu+1}.
\end{align*}
By taking $x=\sqrt{(\nu+1)/(\nu+r^2)}\, w$ with $r=\widetilde{t}_1$, we obtain
\begin{align*}
\mathbb{P}\Biggl(T_{\nu+1}\leqslant \sqrt{\nu+1\over \nu+\widetilde{t_1}^2}\,  w\Biggr)
=
	\mathbb{P}(Z_2 \leqslant w\, \vert \, Z_1=\widetilde{t}_1).
\end{align*}
By differentiating the above identity with respect to $w$, we get
\begin{align*}
\sqrt{\nu+1\over \nu+\widetilde{t_1}^2}\,	f_{\nu+1}\Biggl(\sqrt{\nu+1\over \nu+\widetilde{t_1}^2}\,  w\Biggr)
	=
f_{Z_2}(
w \, \vert\, Z_1=\widetilde{t}_1 ).
\end{align*}
Hence,
\begin{align*}
	\int_{{B}_\rho}		
f_{Z_2}(
w \, \vert\, Z_1=\widetilde{t}_1 )\, 
{\rm d}w
=
\sqrt{\nu+1\over \nu+\widetilde{t_1}^2}
	\int_{{B}_\rho}	
f_{\nu+1}\Biggl(\sqrt{\nu+1\over \nu+\widetilde{t_1}^2}\,  w\Biggr)\,
{\rm d}w
=
F_{\nu+1}\left(\sqrt{\nu+1\over \nu+\widetilde{t_1}^2}\, {B}_\rho \right).	
\end{align*}
Now, by using Theorem \ref{theo-pdf-cond}, the required result follows.
\end{proof}

%

\subsection{Mahalanobis Distance}\label{maha:sec}

The squared Mahalanobis distance of a random vector $\boldsymbol{T}=(T_1,T_2)$ and the vector $\log(\boldsymbol{\eta})=(\log(\eta_1),\log(\eta_2))$ of a bivariate log-symmetric distribution is defined as
\begin{align*}
d^2(\boldsymbol{T},\log(\boldsymbol{\eta}))
=
		{\widetilde{T_1}^2-2\rho\widetilde{T_1}\widetilde{T_2}+\widetilde{T_2}^2
		\over 
		1-\rho^2},
\end{align*}
where
$$
		\widetilde{T_i}
	=
	\log\biggl[\biggl({T_i\over \eta_i}\biggr)^{1/\sigma_i}\biggr], \ \eta_i=\exp(\mu_i), \ i=1,2.
$$

We now derive formulas for the CDF and PDF of the above random variable $d^2(\boldsymbol{T},\log(\boldsymbol{\eta}))$.

\begin{proposition}\label{Mahalanobis Distance}
	If $\boldsymbol{T}\sim {\rm BLS}(\boldsymbol{\theta},g_c)$, then the CDF of $d^2(\boldsymbol{T},\log(\boldsymbol{\eta}))$, denoted by $F_{d^2(\boldsymbol{T},\log(\boldsymbol{\eta}))}$, is given by
	\begin{align}
	F_{d^2(\boldsymbol{T},\log(\boldsymbol{\eta}))}(x)
	&=
	4
	\int_{0}^{\sqrt{x}}
	\biggl[
	F_{Z_2}\Big(\sqrt{x-z_1^2}\,\Big\vert\, Z_1=z_1\Big)-{1\over 2}
	\biggr] 
	f_{Z_1}(z_1)\, {\rm d}z_1 \, \boldsymbol{\cdot} \mathds{1}_{(0,\infty)}(x)
	\label{eq-1}
	\\[0,2cm]
	&=
	{4\over Z_{g_c}}\, 
	\int_{0}^{\sqrt{x}}
	\left[\int_{0}^{\sqrt{x-z_1^2}} g_c(z_1^2+z_2^2) \, {\rm d}z_2\right] {\rm d}z_1\, \boldsymbol{\cdot}  \mathds{1}_{(0,\infty)}(x), 
	\label{eq-2}
	\end{align}
	where $Z_{g_c}$ is as in Proposition \ref{partition function-simplicado}.
\end{proposition}
\begin{proof}
	As $\boldsymbol{T}=(T_1,T_2)$ admits the stochastic representation in Proposition \ref{Stochastic Representation}, there exist $Z_1$ and $Z_2$ such that $\widetilde{T_1}=Z_1$ and $\widetilde{T_2}={\rho} Z_1+\sqrt{1-\rho^2} Z_2$. Then, a simple algebraic manipulation shows that
	\begin{align}\label{identity-dist-mahalanobis}
	d^2(\boldsymbol{T},\log(\boldsymbol{\eta}))
	= Z_1^2+Z_2^2.
	\end{align}
Upon using the law of total expectation, we get (for $x>0$)
\begin{align}
	F_{d^2(\boldsymbol{T},\log(\boldsymbol{\eta}))}(x)
	=
	\mathbb{E}[\mathbb{E}(\mathds{1}_{\{ Z_1^2+Z_2^2\leqslant x \}}\, \vert \, Z_1)]
	&=
	\mathbb{E}\Big[\mathbb{E}\Big(\mathds{1}_{\{\vert Z_2\vert \leqslant\sqrt{x-Z_1^2}\}}\, \Big\vert \, Z_1 \mathds{1}_{\{\vert Z_1\vert \leqslant\sqrt{x}\}}\Big)\Big]
	\nonumber
	\\[0,2cm]
	&=
	\int_{-\sqrt{x}}^{\sqrt{x}}
	\left[\int_{-\sqrt{x-z_1^2}}^{\sqrt{x-z_1^2}} f_{Z_2}(z_2\,\vert\, Z_1=z_1)\, {\rm d}z_2\right] f_{Z_1}(z_1)\, {\rm d}z_1. \label{first-eq}
\end{align}
As $f_{Z_2}(z_2\,\vert\, Z_1=z_1)$ and $f_{Z_1}(z_1)$ are both even functions (see Proposition \ref{joint and marginal pdfs}), the proof of the first equality in \eqref{eq-1} follows from \eqref{first-eq}. The second equality in \eqref{eq-2} follows by using the joint PDF $f_{Z_1,Z_2}$ given in Proposition \ref{joint and marginal pdfs} in \eqref{first-eq}.
\end{proof}

\begin{proposition}\label{Mahalanobis Distance-PDF}
	If $\boldsymbol{T}\sim {\rm BLS}(\boldsymbol{\theta},g_c)$, then the PDF of $d^2(\boldsymbol{T},\log(\boldsymbol{\eta}))$, denoted by $f_{d^2(\boldsymbol{T},\log(\boldsymbol{\eta}))}$, is given by
	\begin{align*}
	f_{d^2(\boldsymbol{T},\log(\boldsymbol{\eta}))}(x)={\pi\over Z_{g_c}}\, g_c(x), \quad x>0,
	\end{align*}
	where $Z_{g_c}$ is as in Proposition \ref{partition function-simplicado}.
\end{proposition}
\begin{proof}
The proof follows immediately by differentiating \eqref{eq-2} with respect to $x$ and then by using the following well-known Leibniz integral rule:
\begin{align*}
	{{\rm d}\over {\rm d}x} \int_{a(x)}^{b(x)} h(x,y)\, {\rm d}y
	=
	h(x,b(x)) b'(x)-h(x,a(x)) a'(x)+ \int_{a(x)}^{b(x)} {\partial h(x,y)\over \partial x}\, {\rm d}y.
\end{align*}
\end{proof}

\begin{remark}
	\begin{itemize}
	\item[(i)]{\it  Gaussian generator.} By taking $g_c(x)=\exp(-x/2)$ and $Z_{g_c}=2\pi$ (see Table \ref{table:1}), and then applying Proposition \ref{Mahalanobis Distance-PDF},  we get
	\begin{align*}
		f_{d^2(\boldsymbol{T},\boldsymbol{\eta})}(x)
	=
{1\over 2}\, \exp\biggl(-{x\over 2}\biggr)
=
{1\over 2^{k/2}\Gamma(k/2)}\, x^{(k/2)-1} \exp\biggl(-{x\over 2}\biggr),
\quad \text{with} \ k=2.
	\end{align*}
But, the formula on the right is the PDF of a random variable following a chi-squared distribution with $k$ degrees of freedom ($\chi^2_k$). Hence, $d^2(\boldsymbol{T},\log(\boldsymbol{\eta}))\sim \chi^2_2$;
	
	\item[(ii)] {\it Student-$t$ generator}. By taking $g_c(x)=(1+(x/\nu))^{-(\nu+2)/2}$ and $Z_{g_c}={{\Gamma({\nu/ 2})}\nu\pi/{\Gamma({(\nu+2)/ 2})}}$ (see Table \ref{table:1}), and then applying Proposition \ref{Mahalanobis Distance-PDF}, we get
		\begin{align*}
	f_{d^2(\boldsymbol{T},\boldsymbol{\eta})}(x)
	&=
	{{\Gamma({(\nu+2)/ 2})}\over {\Gamma({\nu/  2})}\nu}\, \biggl(1+{x\over\nu}\biggr)^{-(\nu+2)/2}
	\\[0,2cm]
	&=
	{1\over 2}\, {\sqrt{[d_1(x/2)]^{d_1}d_2^{d_2}\over [d_1(x/2)+d_2]^{d_1+d_2}}\over (x/2){\rm B}(d_1/2, d_2/2)},
	\quad \text{with} \ d_1=2\ \text{and} \ d_2=\nu.
	\end{align*}
	Here, ${\rm B}(x,y)=\Gamma(x)\Gamma(y)/\Gamma(x+y)$, $x>0, y>0$, is the complete beta function. Observe that the formula on the second identity above is the PDF of a random variable $2X$, where $X$ follows the $F$-distribution with $d_1$ and $d_2$ degrees of freedom ($F_{d_1,d_2}$). Thus, we have $d^2(\boldsymbol{T},\log(\boldsymbol{\eta})) \sim 2 F_{2,\nu}$.
	\end{itemize}
\end{remark}

\subsection{Independence}

\begin{proposition}
Let $\boldsymbol{T}\sim {\rm BLS}(\boldsymbol{\theta},g_c)$. If $\rho=0$ and the density generator $g_c$ in \eqref{PDF} is such that
\begin{align}\label{kernel-dec}
g_c(x^2+y^2)
=g_{c_1}(x^2) 
g_{c_2}(y^2)
\quad \forall (x,y)\in\mathbb{R}^2,
\end{align}
for some density generators $g_{c_1}$ and $g_{c_2}$, then $T_1$ and $T_2$ are statistically independent.
\end{proposition}
\begin{proof}
When $\rho=0$, from \eqref{kernel-dec}, the joint density of $(T_1,T_2)$ is
$$
f_{T_1,T_2}(t_1,t_2;\boldsymbol{\theta})
=
f_1(t_1;\mu_1,\sigma_1)
f_2(t_2;\mu_2,\sigma_2), \quad\forall(t_1,t_2)\in(0,\infty)\times (0,\infty),
$$ 
and consequently $Z_{g_c}=Z_{g_{c_1}}Z_{g_{c_2}}$,
where
\begin{align*}
f_i(t_i;\mu_i,\sigma_i)=
	{1\over t_i\sigma_iZ_{g_{c_i}}}\,
g_{c_i}(\widetilde{t_i}^2), \ t_i>0,
\quad 
\text{and}
\quad
Z_{g_{c_i}}
=
\int_{-\infty}^{\infty}
g_{c_i}({z_i}^2)\, {\rm d}z_i, \quad i=1,2,
\end{align*}
and $\widetilde{t_i}$ is as in \eqref{PDF}.
A simple calculation shows that $f_1$ and $f_2$ are densities functions (in fact, $f_1$ and $f_2$ are densities associated with two univariate continuous symmetric random variables; see \cite{Vanegas2016}. Then, from Proposition 2.5 of \cite{James2004}, it follows that $T_1$ and $T_2$ are independent, and even more, that $f_i=f_{T_i}$, for $i=1,2$.
\end{proof}

\begin{remark}
Observe that, in Table \ref{table:1}, the density generator of the bivariate log-normal is the unique one that satisfies the condition in \eqref{kernel-dec}.
\end{remark}

\subsection{Moments}\label{moments}

\begin{proposition}\label{Real moments}
Let $\boldsymbol{X}=(X_1,X_2)\sim {\rm BES}(\boldsymbol{\theta}_*, g_c)$ and $\boldsymbol{T}=(T_1,T_2)\sim {\rm BLS}(\boldsymbol{\theta},g_c)$.
If the moment-generating function (MGF) of $X_i$, denoted by $M_{X_i}(s_i)$, $i=1,2$, exists, then the moments of $T_i$ are
\begin{align*}
 \mathbb{E}(T_i^r)=\eta_i^r \vartheta(\sigma_i^2r^2), \quad \text{with} \ \eta_i=\exp(\mu_i), \ i=1,2, \ r\in\mathbb{R},
\end{align*}
for some scalar function $\vartheta$, which is called the characteristic generator \cite[see][]{Fang1990}.

For example, when  $g_c(x)=\exp(-x/2)$ (Gaussian generator), $\vartheta(x)=\exp(x/2)$, and when $g_c(x)=(1+(x/\nu))^{-(\nu+2)/2}$, $\nu>0$ (Student-$t$ generator), $\vartheta$ does not exist.
\end{proposition}
\begin{proof}
We only show the case for $i=1$, because the other one follows in a similar way.

As the random variable $X_1$ has a MGF $M_{X_1}(s_1)$, the domain of the characteristic function $\varphi_{X_1}(t)$ can be extended to the complex plane, and
\begin{align*}
M_{X_1}(s_1)= \varphi_{X_1}(-is_1).
\end{align*}
As $X_1=\log(T_1)$, by using the above identity, we immediately get
\begin{align}\label{relation}
\mathbb{E}(T_1^r)
=\mathbb{E}[\exp(rX_1)] \nonumber
=M_{X_1}(r)
&=\exp(r\mu_1) M_{S_1}(\sigma_1r)
\\[0,2cm]
&=\exp(r\mu_1)\varphi_{S_1}(-i\sigma_1r)
=\exp(r\mu_1)\varphi_{S_1,S_2}(-i\sigma_1r,0),
\end{align}
with $(S_1,S_2)\sim {\rm BES}(\boldsymbol{\theta}_{*_0}, g_c)$, $\boldsymbol{\theta}_{*_0}=(0,0,1,1,\rho)$, and $\varphi_{S_1,S_2}(s_1,0)$
is the marginal characteristic function. On the other hand, the characteristic function of the BES distribution  is given by \cite[see Item 13.10, p. 595 of][]{Bala2009}
\begin{align}\label{CF}
	\varphi_{S_1,S_2}(s_1,s_2)=\vartheta(s_1^2+2\rho s_1s_2+s_2^2),
\end{align}
where $\vartheta$ is the characteristic generator specified in the statement of the proposition.
Finally, by using \eqref{CF} on the right-hand side of \eqref{relation}, the required result follows.
\end{proof}

\subsection{Correlation function}\label{Correlation function}

By using the stochastic representation (in Proposition \ref{Stochastic Representation}) of $\boldsymbol{T}=(T_1,T_2)$ and the law of total expectation, we have
\begin{align*}
\mathbb{E}(T_1T_2)
&=
\mathbb{E}(\mathbb{E}(T_1T_2\,\vert\, T_1))
=
\mathbb{E}(T_1\mathbb{E}(T_2\,\vert\, T_1))
\\[0,2cm]
&=
\eta_1
\eta_2
\mathbb{E} 
\bigl[
\exp\big((\sigma_1+\sigma_2 {\rho}) Z_1\big)\,
\mathbb{E} 
\big(
\exp(\sigma_2\sqrt{1-\rho^2} Z_2)
\,\big\vert\, Z_1 
\big)
\bigr],
\end{align*}
where $Z_1$ and $Z_2$ are as defined in Proposition \ref{Stochastic Representation}.

Hence, from the formula of moments (in Proposition \ref{Real moments}), we obtain the correlation function of $T_1$ and $T_2$ to be
\begin{align*}
	\rho(T_1,T_2)
	=
	\dfrac{
		\mathbb{E} 
		\bigl[
		\exp\big((\sigma_1+\sigma_2 {\rho}) Z_1\big)\,
		\mathbb{E} 
		\big(
		\exp(\sigma_2\sqrt{1-\rho^2} Z_2)
		\,\big\vert\, Z_1 
		\big)
		\bigr]-\vartheta(\sigma_1^2)\vartheta(\sigma_2^2)}
	{\sqrt{\vartheta(4\sigma_1^2)-\vartheta^2(\sigma_1^2)}\, \sqrt{\vartheta(4\sigma_2^2)-\vartheta^2(\sigma_2^2)}},
\end{align*}
where $\vartheta$ is a scalar function as stated in Proposition \ref{Real moments}.

It is a simple task to verify that
when  $g_c(x)=\exp(-x/2)$ (Gaussian generator), $\rho(T_1,T_2)=[\exp(\sigma_1\sigma_2\rho)-1]/\big[\sqrt{\exp(\sigma_1^2)-1} \, \sqrt{\exp(\sigma_2^2)-1}\,\big]$, and when $g_c(x)=(1+(x/\nu))^{-(\nu+2)/2}$, $\nu>0$ (Student-$t$ generator), $\rho(T_1,T_2)$ does not exist.

\subsection{Some other properties}

If $\boldsymbol{T}=(T_1,T_2)\sim {\rm BLS}(\boldsymbol{\theta},g_c)$, then the following properties follow immediately as a consequence of the definition of the BLS distribution, analogous to the properties stated by \cite{Vanegas2016}:
\begin{itemize}
	\item[(P1)]  The CDF of $\boldsymbol{T}$ is given by
	$F_{T_1,T_2}(t_1,t_2;\boldsymbol{\theta})
	=F_{S_1,S_2}(\widetilde{t}_1,\widetilde{t}_2;\boldsymbol{\theta}_{*_0})$, with $(S_1,S_2)\sim {\rm BES}(\boldsymbol{\theta}_{*_0}, g_c)$ and $\boldsymbol{\theta}_{*_0}=(0,0,1,1,\rho)$;
	\item[(P2)]  The random vector $(T_1^*,T_2^*)=([T_1/\eta_1]^{1/\sigma_1},[T_2/\eta_2]^{1/\sigma_2})$ follows standard BLS distribution. In other words, $(T_1^*,T_2^*)\sim {\rm BLS}(\boldsymbol{\theta}_0, g_c)$ with $\boldsymbol{\theta}_0=(1,1,1,1,\rho)$;
	\item[(P3)] $(c_1T_1,c_2T_2)\sim {\rm BLS}(c_1\eta_1,c_2\eta_2,\sigma_1,\sigma_2,g_c)$ for all constants $c_1,c_2>0$;
	\item[(P4)] $(T_1^{c_1},T_2^{c_2})\sim {\rm BLS}(\eta_1^{c_1},\eta_2^{c_2},c_1^2\sigma_1,c_2^2\sigma_2,g_c)$ for all constants $c_1\neq 0$ and $c_2\neq 0$.
\end{itemize}

\begin{proposition}
	If $\boldsymbol{T}=(T_1,T_2)\sim {\rm BLS}(\boldsymbol{\theta},g_c)$, then
	the random vectors $({\eta_1/ T_1},{\eta_2/ T_2})$ and $({T_1/ \eta_1},{T_2/ \eta_2})$ are identically distributed. Furthermore,
	$({\eta_1/ T_1},{\eta_2/ T_2})\sim {\rm BLS}(\boldsymbol{\theta}_{\bullet}, g_c)$ and $({T_1/ \eta_1},{T_2/ \eta_2})\sim {\rm BLS}(\boldsymbol{\theta}_{\bullet}, g_c)$, with $\boldsymbol{\theta}_{\bullet}=(1,1,\sigma_1,\sigma_2,\rho)$.
\end{proposition}
\begin{proof}	
	By using the well-known identity that \cite[see e.g. p. 59 of][]{James2004}
	\begin{align*}
	\mathbb{P}(a_1 < X \leqslant b_1, a_2 < Y \leqslant b_2) = F_{X,Y}(b_1, b_2) - F_{X,Y}(b_1, a_2) - F_{X,Y}(a_1, b_2) + F_{X,Y}(a_1, a_2)
	\end{align*}
for any two random variables $X$ and $Y$ and for all $a_1<b_1$ and $a_2<b_2$, with $a_1=\eta_1/w_1$, $b_1=\infty$, $a_2=\eta_2/w_2$ and $b_2=\infty$, for all $(w_1,w_2)\in (0,\infty)^2$, we obtain
	\begin{align}\label{P1}
	\mathbb{P}\Big({\eta_1\over T_1} \leqslant w_1, {\eta_2\over T_2} \leqslant w_2\Big) 
	=
	1- F_{T_2}\Big({\eta_2\over w_2}\Big) -F_{T_1}\Big({\eta_1\over w_1}\Big)+F_{T_1,T_2}\Big({\eta_1\over w_1},{\eta_2\over w_2}\Big).
	\end{align}
	
   Because
	\begin{align}\label{P2}
	\mathbb{P}\biggl({T_1\over \eta_1} \leqslant w_1, {T_2\over \eta_2} \leqslant w_2\biggr) 
	=
	F_{T_1,T_2}({\eta_1 w_1},{\eta_2 w_2}),
	\end{align}
	by taking partial derivatives with respect to $w_1$ and $w_2$ in \eqref{P1}  and \eqref{P2}, we have the joint PDF of ${\eta_1/ T_1}$ and ${\eta_2/ T_2}$,  and  the joint PDF of ${T_1/ \eta_1}$ and ${T_2/ \eta_2}$, to be related as follows:
	\begin{align*}
	f_{{\eta_1\over T_1}, {\eta_2\over T_2}}(w_1,w_2)
	=
	f_{{T_1\over \eta_1}, {T_2\over \eta_2}}(w_1,w_2)	
	= 
	{1\over w_1w_2\sigma_1\sigma_2\sqrt{1-\rho^2}Z_{g_c}}\,
	g_c\Biggl(
	{\widetilde{w}_1^2-2\rho\widetilde{w}_1\widetilde{w}_2+\widetilde{w}_2^2
		\over 
		1-\rho^2}
	\Biggr),
	\quad 
	w_1,w_2>0,
	\end{align*}
	with  $\widetilde{w_i}=\log({w_i}^{1/\sigma_i})$, $i=1,2$. This completes the proof of the proposition.
\end{proof}

\section{Maximum likelihood estimation} \label{Sec:4}
\noindent
Let $\{(T_{1i},T_{2i}):i=1,\ldots,n\}$ be a bivariate random sample of size $n$ from the ${\rm BLS}(\boldsymbol{\theta},g_c)$ distribution with PDF as given in \eqref{PDF}, and let $(t_{1i},t_{2i})$ be the corresponding sample of $(T_{1i},T_{2i})$. Then, the log-likelihood function for $\boldsymbol{\theta}=(\eta_1,\eta_2,\sigma_1,\sigma_2,\rho)$,
without the additive constant, is given by
\begin{align*}
\ell(\boldsymbol{\theta})
=
-n\log(\sigma_1)-n\log(\sigma_2)-{n\over 2}\,\log\big({1-\rho^2}\big)
+
\sum_{i=1}^{n}
\log	
g_c\Biggl(
{\widetilde{t_{1i}}^2-2\rho\widetilde{t}_{1i}\widetilde{t}_{2i}+\widetilde{t_{2i}}^2
	\over 
	1-\rho^2}
\Biggr), \quad 
t_{1i},t_{2i}>0,
\end{align*}
where
$$ 
\widetilde{t}_{ki}=\log\biggl[\Bigl({t_{ki}\over \eta_k}\Bigr)^{1/\sigma_k}\biggr], \ \eta_k=\exp(\mu_k), \ k=1,2; \ i=1,\ldots,n.
$$

In the case when a supremum $\widehat{\boldsymbol{\theta}}=(\widehat{\eta_1},\widehat{\eta_2},\widehat{\sigma_1},\widehat{\sigma_2},\widehat{\rho})$ exists, it must satisfy the likelihood equations
\begin{align}\label{likelihood equation}
{\partial \ell({\boldsymbol{\theta}})\over\partial\eta_1}
\bigg\vert_{{\boldsymbol{\theta}}=\widehat{\boldsymbol{\theta}}}
=0,
\quad 
{\partial \ell(\boldsymbol{\theta})\over\partial\eta_2}=0,
\quad 
{\partial\ell(\boldsymbol{\theta})\over\partial\sigma_1}
\bigg\vert_{{\boldsymbol{\theta}}=\widehat{\boldsymbol{\theta}}}=0,
\quad 
{\partial\ell(\boldsymbol{\theta})\over\partial\sigma_2}
\bigg\vert_{{\boldsymbol{\theta}}=\widehat{\boldsymbol{\theta}}}=0,
\quad 
{\partial\ell(\boldsymbol{\theta})\over\partial\rho}
\bigg\vert_{{\boldsymbol{\theta}}=\widehat{\boldsymbol{\theta}}}=0,
\end{align}
where
\begin{align}
	&{\partial \ell(\boldsymbol{\theta})\over\partial\eta_1}
	=
	\frac{2}{\sigma_1\eta_1(1-\rho^2)}
	\sum_{i=1}^{n}
\big(\rho \widetilde{t}_{2i} -\widetilde{t}_{1i}\big)
G(\widetilde{t}_{1i},\widetilde{t}_{2i})
, \nonumber
	\\[0,1cm]
&{\partial \ell(\boldsymbol{\theta})\over\partial\eta_2}
=
\frac{2}{\sigma_2\eta_2(1-\rho^2)}
\sum_{i=1}^{n}
\big(\rho \widetilde{t}_{1i} -\widetilde{t}_{2i}\big)
G(\widetilde{t}_{1i},\widetilde{t}_{2i})
, \nonumber
	\\[0,1cm]
	&{\partial\ell(\boldsymbol{\theta})\over\partial\sigma_1}
	=
	-\frac{n}{\sigma_1}
	+
	{2\over \sigma_1(1-\rho^2)}
		\sum_{i=1}^{n}
\widetilde{t}_{1i}
\big(\rho\widetilde{t}_{2i}-\widetilde{t}_{1i}\big)
G(\widetilde{t}_{1i},\widetilde{t}_{2i})
, \nonumber
	\\[0,1cm]
&{\partial\ell(\boldsymbol{\theta})\over\partial\sigma_2}
=
-\frac{n}{\sigma_2}
+
{2\over \sigma_2(1-\rho^2)}
\sum_{i=1}^{n}
\widetilde{t}_{2i}
\big(\rho\widetilde{t}_{1i}-\widetilde{t}_{2i}\big)
G(\widetilde{t}_{1i},\widetilde{t}_{2i})
, \nonumber
	\\[0,1cm]
	&{\partial \ell(\boldsymbol{\theta})\over\partial\rho}
	=
	{n\rho\over 1-\rho^2}
	-
	{2\over (1-\rho^2)^2}
	\sum_{i=1}^{n}
\big(\rho \widetilde{t}_{1i}-\widetilde{t}_{2i}\big)
\big(\rho\widetilde{t}_{2i}-\widetilde{t}_{1i}\big)
%
G(\widetilde{t}_{1i},\widetilde{t}_{2i}), \label{rho-mle}
\end{align}
with the notation
\begin{align*}
&G(\widetilde{t}_{1i},\widetilde{t}_{2i})
=
g_c'\Biggl(
	{\widetilde{t_{1i}}^2-2\rho\widetilde{t}_{1i}\widetilde{t}_{2i}+\widetilde{t_{2i}}^2
		\over 
		1-\rho^2}
	\Biggr)
	{\Bigg /}
	g_c\Biggl(
	{\widetilde{t_{1i}}^2-2\rho\widetilde{t}_{1i}\widetilde{t}_{2i}+\widetilde{t_{2i}}^2
		\over 
		1-\rho^2}
	\Biggr),
	\quad i=1,\ldots,n.
\end{align*}
A simple observation shows that the likelihood equations in \eqref{likelihood equation} can be rewritten as follows:
\begin{align*}
&\sum_{i=1}^{n}
\widetilde{t}_{1i}\,
G(\widetilde{t}_{1i},\widetilde{t}_{2i})
\bigg\vert_{{\boldsymbol{\theta}}=\widehat{\boldsymbol{\theta}}}
=0,
\\[0,1cm]
&
\sum_{i=1}^{n}
\big(\widetilde{t}_{1i}^2-\widetilde{t}_{2i}^2\big)\,
G(\widetilde{t}_{1i},\widetilde{t}_{2i})
\bigg\vert_{{\boldsymbol{\theta}}=\widehat{\boldsymbol{\theta}}}
=0,
\\[0,2cm]
&
\sum_{i=1}^{n}
\widetilde{t_{2i}}
\left[2\rho
\widetilde{t_{2i}}
-
(1+\rho^2) \widetilde{t}_{1i}\right]
G(\widetilde{t}_{1i},\widetilde{t}_{2i})
\bigg\vert_{{\boldsymbol{\theta}}=\widehat{\boldsymbol{\theta}}}
=-{n\widehat{\rho}(1-\widehat{\rho}^2)\over 2}\,.
\end{align*}
Any nontrivial root $\widehat{\boldsymbol{\theta}}$ of the above likelihood equations is known as
an ML estimator in the loose sense. When the parameter value provides the absolute maximum of the log-likelihood function, it is called an ML estimator in the strict sense.

In the following proposition, we discuss the existence of the ML estimator $\widehat{\rho}$ when the other parameters are all known.
\begin{proposition}\label{prop-existence-MLE}
Let $g_c$ be a density generator satisfying the following condition: 
\begin{align}\label{condition-g}
g'_c(x)=r(x) g_c(x),\quad -\infty<x<\infty,
\end{align}
for some real-valued function $r(x)$ such that $\lim_{\rho\to \pm 1} r(x_{\rho,i})=c\in(-\infty,0)$, where $x_{\rho,i}= (\widetilde{t_{1i}}^2-2\rho\widetilde{t}_{1i}\widetilde{t}_{2i}+\widetilde{t_{2i}}^2)/(1-\rho^2)$, $i=1,\ldots,n$.
If the parameters $\eta_1,\eta_2,\sigma_1$ and $\sigma_2$ are all known, then \eqref{rho-mle} has at least one root in the interval $(-1, 1)$.

\end{proposition}
\begin{proof}
As $g'_c(x_{\rho,i})=r(x_{\rho,i}) g_c(x_{\rho,i})$, we have
$G(\widetilde{t}_{1i},\widetilde{t}_{2i})=r(x_{\rho,i})$. Then, by using the condition that $\lim_{\rho\to \pm 1} r(x_{\rho,i})=c<0$, from \eqref{rho-mle}, we see that
\begin{align*}
\lim_{\rho\to 1^-}
{\partial \ell(\boldsymbol{\theta})\over\partial\rho}
=-\infty
\quad \text{and} \quad
\lim_{\rho\to -1^+}
{\partial \ell(\boldsymbol{\theta})\over\partial\rho}
=+\infty.
\end{align*}
Hence, by Intermediate value theorem, there
exists at least one solution in the interval $(-1, 1)$.
\end{proof}
\begin{remark}
Observe that in Table \ref{table:1}, the density generators
of the Bivariate Log-normal (or Bivariate Log-power-exponential with $\xi= 0$) and  the Bivariate Log-Kotz type (with $\delta=1$) 
satisfy the hypotheses of Proposition \ref{prop-existence-MLE} with $r(x_{\rho,i})=-1/2$ and $r(x_{\rho,i})=(-\lambda x_{\rho,i}+\xi-1)/x_{\rho,i} {\longrightarrow}-\lambda$ as $\rho\to\pm 1$, $\forall i=1,\ldots,n$, respectively. Then, Proposition \ref{prop-existence-MLE} can be applied to guarantee the existence of an ML estimator $\widehat{\rho}$ of $\rho$ in the loose sense.

On the other hand, the  density generators of the Bivariate Log-Kotz type (with $\delta<1$), 
the Bivariate Log-Student-$t$, 
the Bivariate Log-Pearson Type VII and
the Bivariate Log-power-exponential (with $\xi\neq 0$)
%
%
satisfy the condition \eqref{condition-g} with
$r(x)=(-\lambda\delta x^\delta+\xi-1)/x$,
$r(x)=-(\nu+2)/2(1+{x\over\nu})$, 
$r(x)=-\xi/(\theta+x)$
and
$r(x)=-x^{-\xi/(\xi+1)}/2(\xi+1)$,  
respectively, but in all these cases $ r(x_{\rho,i})\longrightarrow 0$ as $\rho\to\pm 1$, $\forall i=1,\ldots,n$.
\end{remark}

For the BLS model, no closed-form solution to the maximization problem is available, and an MLE can only be found via numerical optimization. Under mild regularity conditions \citep{Cox1974,Davison2008}, the asymptotic distribution of ML estimator  $\widehat{\boldsymbol{\theta}}$ of $\boldsymbol{\theta}$ is easily determined by the convergence in law: $(\widehat{\boldsymbol{\theta}}-\boldsymbol{\theta})\stackrel{\mathscr D}{\longrightarrow} N(\boldsymbol{0},I^{-1}(\boldsymbol{\theta}))$,
where  
$\boldsymbol{0}$ is the zero mean vector and $I^{-1}(\boldsymbol{\theta})$ is the inverse expected 
Fisher information matrix.
The main use of the last convergence is to construct confidence regions and to perform
hypothesis testing for $\boldsymbol{\theta}$ \citep{Davison2008}. 


\section{Simulation study}\label{Sec:5}

In this section, we carry out a Monte Carlo (MC) simulation study to evaluate the performance of the above described ML estimators for the BLS models. We use different sample sizes and parameter settings, using log-slash, log-power-exponential and log-normal distributions.

The simulation scenario considers the following setting: 1,000 MC replications, sample size $n \in (25,50,100,150)$, vector of true parameters $(\eta_1,\eta_2,\sigma_1,\sigma_2)= (1,1,0.5,0.5)$, $\rho \in \{0,0.25,0.5,0.75,0.95\}$, $\nu = 4$ (log-slash), and $\xi = 0.3$ (log-power-exponential). The extra parameters of the chosen distributions are assumed to be fixed; see \cite{Sauloetal2022} for details.

The performance and recovery of the ML estimators were evaluated through the empirical bias and the mean square error (MSE), which are calculated from the MC replicates, as
\begin{align}
 \text{Bias}(\widehat{\theta}) = \frac{1}{N}\sum_{i=1}^N\widehat{\theta}^{(i)} - \theta,
 \quad
  \quad
 \text{MSE}(\widehat{\theta}) = \frac{1}{N}\sum_{i=1}^N(\widehat{\theta}^{(i)} - \theta)^2,
\end{align}
where $\theta$ and $\widehat{\theta}^{(i)}$ are the true value of the parameter and its respective $i$th estimate, and $N$ is the number of MC replications. The steps for the MC simulation study are described in Algorithm 1 below:


\begin{table}[H]
\resizebox{\linewidth}{!}{
\begin{tabular}{l}
\hline
\textbf{Algorithm 1.} Simulation                                               \\ \hline
1. Choose the BLS distribution based on Table 1 and define the value of the parameters of the chosen distribution; \\
2. Generate 1,000 samples of size $n$ from the chosen model; \\
3. Estimate the model parameters using the ML method for each sample; \\
4. Compute the empirical bias and MSE.                                                                           \\ \hline
\end{tabular}}
\end{table}

The simulation results are presented in Tables \ref{table:2}-\ref{table:4}. We observe that the results obtained for the chosen distributions are as expected. As the sample size increases, the bias and MSE tend to decrease. Moreover, in general, the results do not seem to depend on the parameter $\rho$.

\begin{table}[htpb!]\label{Log-Slash}
\caption{Monte Carlo simulation results for the bivariate log-slash distribution with $\nu = 4$.}
\resizebox{\linewidth}{!}{
\begin{tabular}{llcccccccccccccc}
\hline
\multirow{2}{*}{$n$} & \multirow{2}{*}{$\rho$} & \multicolumn{10}{l}{MLE}                                                                                                                                   \\ \cline{3-12}
                   &                         & \multicolumn{2}{c}{$\widehat{\eta}_1$} & \multicolumn{2}{c}{$\widehat{\eta}_2$} & \multicolumn{2}{c}{$\widehat{\sigma}_1$} & \multicolumn{2}{c}{$\widehat{\sigma}_2$} & \multicolumn{2}{c}{$\widehat{\rho}$} \\
                   \noalign{\hrule height 1.7pt}
                   &                         & Bias          & MSE          & Bias          & MSE          & Bias            & MSE          & Bias           & MSE           & Bias          & MSE        \\
25                 & 0.0                     & 0.0093   & 0.0138  & 0.0085   & 0.0145      & -0.1040         & 0.0154       & -0.1004   & 0.0143  & -0.0053       & 0.0426     \\
                   & 0.25                    & 0.0078   & 0.0137  & 0.0096   & 0.0150       & -0.1053         & 0.0148       & -0.0988   & 0.0140  & -0.0079       & 0.0412     \\
                   & 0.50                    & 0.0050   & 0.0126  & 0.0060   & 0.0133       & -0.1008         & 0.0142       & -0.1046   & 0.0149  & -0.0145       & 0.0276     \\
                   & 0.95                    & 0.0072   & 0.0150  & 0.0064   & 0.0148       & -0.1055         & 0.0151       & -0.1048   & 0.0150  & -0.0046       & 0.0007     \\ \hline
50                 & 0.0                     & 0.0056   & 0.0067  & 0.0030   & 0.0071       & -0.0967        & 0.0113       & -0.0947   & 0.0110  & -0.0065      & 0.0230     \\
                   & 0.25                    & 0.0053   & 0.0069  & 0.0017   & 0.0072       & -0.0963         & 0.0114       & -0.0968   & 0.0114  & -0.0071       & 0.0230     \\
                   & 0.50                    & 0.0033   & 0.0065  & 0.0041   & 0.0073       & -0.0959         & 0.0112       & -0.0951   & 0.0111  & -0.0013       & 0.0130     \\
                   & 0.95                    & 0.0060   & 0.0071  & 0.0066   & 0.0070       & -0.0944         & 0.0109       & -0.0938   & 0.0108  & -0.0003       & 0.0002     \\ \hline
100                & 0.0                     & 0.0006   & 0.0037  & 0.0023   & 0.0034       & -0.0955         & 0.0101       & -0.0944   & 0.0098  & 0.0015        & 0.0122     \\
                   & 0.25                    & -0.0016  & 0.0033  & -0.0007  & 0.0035       & -0.0929         & 0.0096       & -0.0943   & 0.0098  & -0.0056       & 0.0095     \\
                   & 0.50                    & 0.0003   & 0.0035  & 0.0023   & 0.0034       & -0.0949         & 0.0100       & -0.0962   & 0.0102  & -0.0047       & 0.0066     \\
                   & 0.95                    & 0.0009   & 0.0036  & 0.0015   & 0.0035       & -0.0955         & 0.0101       & -0.0960   & 0.0102  & -0.0012       & 0.0001     \\ \hline
150                & 0.0                     & 0.0034   & 0.0022  & 0.0033   & 0.0024       & -0.0953         & 0.0098       & -0.0941   & 0.0095  & 0.0050        & 0.0075     \\
                   & 0.25                    & 0.0000   & 0.0023  & -0.0004  & 0.0022       & -0.0933         & 0.0093       & -0.0934   & 0.0094  & -0.0045       & 0.0066     \\
                   & 0.50                    & 0.0027   & 0.0022  & 0.0038   & 0.0023       & -0.0942         & 0.0095       & -0.0932   & 0.0094  & -0.0015       & 0.0039     \\
                   & 0.95                    & 0.0007   & 0.0023  & 0.0011  & 0.0023       & -0.0936         & 0,0094       & -0.0937   & 0.0094  & -0.0006       & 0.0001     \\ \hline
\end{tabular}
}
\label{table:2}
\end{table}

\begin{table}[htpb]
\caption{Monte Carlo simulation results for the bivariate log-power-exponential distribution with $\xi = 0.3$.}
\resizebox{\linewidth}{!}{
\begin{tabular}{llcccccccccc}
\hline
\multicolumn{1}{c}{\multirow{2}{*}{$n$}} & \multirow{2}{*}{$\rho$} & MLE                      &                         &                          &                         &                          &                         &                          &                         &                             &                         \\ \cline{3-12}
\multicolumn{1}{c}{}                   &                         & \multicolumn{2}{c}{$\widehat{\eta}_1$}                       & \multicolumn{2}{c}{$\widehat{\eta}_2$}                       & \multicolumn{2}{c}{$\widehat{\sigma}_1$}                     & \multicolumn{2}{c}{$\widehat{\sigma}_2$}                     & \multicolumn{2}{c}{$\widehat{\rho}$}                               \\ \hline
                                       &                         & \multicolumn{1}{c}{Bias} & \multicolumn{1}{c}{MSE} & \multicolumn{1}{c}{Bias} & \multicolumn{1}{c}{MSE} & \multicolumn{1}{c}{Bias} & \multicolumn{1}{c}{MSE} & \multicolumn{1}{c}{Bias} & \multicolumn{1}{c}{MSE} & \multicolumn{1}{c}{Bias}    & \multicolumn{1}{c}{MSE} \\
\multirow{4}{*}{25}                    & 0.00                    & 0.0076                   & 0.0174                  & 0.0133                   & 0.0173                  & -0.0466                  & 0.0138                  & -0.0488                  & 0.0139                  & \multicolumn{1}{c}{-0.0038} & 0.0411                  \\
                                       & 0.25                    & 0.0048                   & 0.0155                  & 0.0076                   & 0.0169                  & -0.0434                  & 0.0079                  & -0.0432                  & 0.0082                  & \multicolumn{1}{c}{-0.0098} & 0.0399                  \\
                                       & 0.50                    & 0.0066                   & 0.0170                  & 0.0099                   & 0.0192                  & -0.0436                  & 0.0097                  & -0.0413                  & 0.0101                  & \multicolumn{1}{c}{-0.0047} & 0.0268                  \\
                                       & 0.95                    & 0.0108                   & 0.0170                  & 0.0113                   & 0.0171                  & -0.0423                  & 0.0095                  & -0.0431                  & 0.0097                  & -0.0025                     & 0.0006                  \\ \hline
\multirow{4}{*}{50}                    & 0.00                    & 0.0026                   & 0.0083                  & 0.0060                   & 0.0084                  & -0.0735                  & 0.0415                  & -0.0712                  & 0.0418                  & \multicolumn{1}{c}{0.0016}  & 0.0192                  \\
                                       & 0.25                    & 0.0062                   & 0.0090                  & 0.0017                   & 0.0081                  & -0.0572                  & 0.0243                  & -0.0597                  & 0.0258                  & 0.0003                      & 0.0181                  \\
                                       & 0.50                    & 0.0055                   & 0.0085                  & 0.0027                   & 0.0076                  & -0.0452                  & 0.0128                  & -0.0460                  & 0.0125                  & -0.0026                     & 0.0134                  \\
                                       & 0.95                    & 0.0062                   & 0.0086                  & 0.0068                   & 0.0088                  & -0.0364                  & 0.0047                  & -0.0358                  & 0.0047                  & -0.0010                     & 0.0003                  \\ \hline
\multirow{4}{*}{100}                   & 0.00                    & 0.0034                   & 0.0044                  & 0.0022                   & 0.0043                  & -0.0475                  & 0.0173                  & -0.0453                  & 0.0175                  & 0.0011                      & 0.0107                  \\
                                       & 0.25                    & 0.0023                   & 0.0041                  & 0.0025                   & 0.0042                  & -0.0331                  & 0.0044                  & -0.0319                  & 0.0044                  & 0.0011                      & 0.0088                  \\
                                       & 0.50                    & 0.0011                   & 0.0045                  & 0.0000                   & 0.0041                  & -0.0313                  & 0.0035                  & -0.0345                  & 0.0035                  & -0.0058                     & 0.0063                  \\
                                       & 0.95                    & 0.0034                   & 0.0043                  & 0.0027                   & 0.0043                  & -0.0344                  & 0.0041                  & -0.0344                  & 0.0040                  & -0.0009                     & 0.0001                  \\ \hline
\multirow{4}{*}{150}                   & 0.00                    & -0.0002                  & 0.0029                  & 0.0038                   & 0.0029                  & -0.0319                  & 0.0020                  & -0.0318                  & 0.0020                  & -0.0008                     & 0.0066                  \\
                                       & 0.25                    & 0.0001                   & 0.0028                  & 0.0036                   & 0.0026                  & -0.0300                  & 0.0019                  & -0.0303                  & 0.0019                  & -0.0018                     & 0.0067                  \\
                                       & 0.50                    & 0.0050                   & 0.0026                  & 0.0027                   & 0.0027                  & -0.0331                  & 0.0036                  & -0.0347                  & 0.0037                  & -0.0035                     & 0.0041                  \\
                                       & 0.95                    & 0.0014                   & 0.0029                  & 0.0014                   & 0.0029                  & -0.0306                  & 0.0019                  & -0.0299                  & 0.0018                  & -0.0002                     & 0.0001                  \\ \hline
\end{tabular}
}
\label{table:3}
\end{table}

\begin{table}[htpb!]
\caption{Monte Carlo simulation results for the bivariate log-normal distribution.}
\resizebox{\linewidth}{!}{
\begin{tabular}{llcccccccccc}
\hline
\multicolumn{1}{c}{\multirow{2}{*}{$n$}} & \multirow{2}{*}{$\rho$} & \multicolumn{1}{l}{MLE} &        &         & \multicolumn{1}{l}{} &         & \multicolumn{1}{l}{} & \multicolumn{1}{l}{} & \multicolumn{1}{l}{} & \multicolumn{1}{l}{} & \multicolumn{1}{l}{} \\ \cline{3-12}
\multicolumn{1}{c}{}                   &                         & \multicolumn{2}{c}{$\widehat{\eta}_1$}     & \multicolumn{2}{c}{$\widehat{\eta}_2$}   & \multicolumn{2}{c}{$\widehat{\sigma}_1$} & \multicolumn{2}{c}{$\widehat{\sigma}_2$}              & \multicolumn{2}{c}{$\widehat{\rho}$}                     \\ \hline
                                       &                         & Bias                    & MSE    & Bias    & MSE                  & Bias    & MSE                  & Bias                 & MSE                  & Bias                & MSE                  \\
\multirow{4}{*}{25}                    & 0.00                    & -0.0007                 & 0.0100 & 0.0034  & 0.0100               & -0.0312 & 0.0182               & -0.0265              & 0.0178               & -0.0044              & 0.0411               \\
                                       & 0.25                    & 0.0095                  & 0.0107 & 0.0085  & 0.0101               & -0.0336 & 0.0180               & -0.0320              & 0.0220               & -0.0151              & 0.0367               \\
                                       & 0.50                    & 0.0120                  & 0.0109 & 0.0159  & 0.0112               & -0.0313 & 0.0205               & -0.0282              & 0.0198               & -0.0106              & 0.0252               \\
                                       & 0.95                    & -0.0013                 & 0.0101 & -0.0011 & 0.0100               & -0.0289 & 0.0197               & -0.0296              & 0.0199               & -0.0026              & 0.0005               \\ \hline
\multirow{4}{*}{50}                    & 0.00                    & -0.0012                 & 0.0051 & 0.0055  & 0.0050               & -0.0394 & 0.0373               & -0.0433              & 0.0387               & 0.0006               & 0.0187               \\
                                       & 0.25                    & 0.0014                  & 0.0053 & 0.0023  & 0.0050               & -0.0445 & 0.0391               & -0.0446              & 0.0387               & -0.0014              & 0.0174               \\
                                       & 0.50                    & 0.0023                  & 0.0051 & -0.0027 & 0.0052               & -0.0370 & 0.0329               & -0.0363              & 0.0338               & 0.0044               & 0.0114               \\
                                       & 0.95                    & 0.0020                  & 0.0051 & 0.0026  & 0.0051               & -0.0151 & 0.0127               & -0.0148              & 0.0125               & 0.0000               & 0.0002               \\ \hline
\multirow{4}{*}{100}                   & 0.00                    & 0.0021                  & 0.0025 & 0.0020  & 0.0024               & -0.0043 & 0.0022               & -0.0067              & 0.0021               & 0.0079               & 0.0096               \\
                                       & 0.25                    & -0.0025                 & 0.0025 & -0.0031 & 0.0025               & -0.0063 & 0.0042               & -0.0073              & 0.0039               & 0.0008               & 0.0091               \\
                                       & 0.50                    & -0.0019                 & 0.0025 & -0.0007 & 0.0023               & -0.0035 & 0.0013               & -0.0052              & 0.0013               & 0.0001               & 0.0057               \\
                                       & 0.95                    & 0.0025                  & 0.0024 & 0.0029  & 0.0024               & -0.0061 & 0.0038               & -0.0063              & 0.0039               & -0.0001              & 0.0001               \\ \hline
\multirow{4}{*}{150}                   & 0.00                    & 0.0000                  & 0.0018 & -0.0018 & 0.0016               & -0.0101 & 0.0080               & -0.0096              & 0.0081               & -0.0024              & 0.0064               \\
                                       & 0.25                    & 0.0003                  & 0.0017 & 0.0001  & 0.0017               & -0.0041 & 0.0027               & -0.0035              & 0.0027               & 0.0016               & 0.0058               \\
                                       & 0.50                    & 0.0000                  & 0.0016 & -0.0001 & 0.0017               & -0.0088 & 0.0066               & -0.0082              & 0.0064               & -0.0021              & 0.0040               \\
                                       & 0.95                    & 0.0013                  & 0.0018 & 0.0016  & 0.0018               & -0.0115 & 0.0100               & -0.0118              & 0.0100               & -0.0006              & 0.0001               \\ \hline
\end{tabular}
}
\label{table:4}
\end{table}
\newpage

\section{Application to fatigue data analysis}\label{Sec:6}

In this section, we illustrate the proposed model and the inferential method using fatigue. The data are from \cite{Marchantetal2016}, in which they used a multivariate Birnbaum-Saunders regression model to represent the fatigue data. The fatigue concerns the process of material failure, caused by cyclic stress. Thus, fatigue is composed of crack initiation and propagation, until the material fractures. The calculation of fatigue life is important for determining the reliability of components or structures. Here, we consider the variables \textit{Von Mises stress} $(T_1$, \text{in} $N/mm^2)$ and \textit{die lifetime} ($T_2$, in number of cycles). According to \cite{Marchantetal2016}, die fracture is the fatigue of metal caused by cyclic stress in the course of the service life cycle of dies (die lifetime).

Table \ref{table:5} provides descriptive statistics for the variables \textit{Von Mises stress} ($T_1$) and \textit{die lifetime} ($T_2$), including the minimum, median, mean, maximum, standard deviation (SD), coefficient of variation, coefficient of skewness (CS), and coefficient of kurtosis (CK). We observe in the variable \textit{Von mises stress}, the mean and median to be, respectively, $1.247$ and $1.130$, i.e., the mean is larger than the median, which indicates a positive skewness in the data. The CV is $56.172\%$, which means a moderate level of dispersion around the mean. Furthermore, the CS value also confirms the skewed nature. The variable \textit{die lifetime} has mean equal to $23.761$ and median equal to $19.000$, which also indicate the positively skewed feature in the distribution of the data. Moreover, the CV value is $71.967\%$, showing a moderate level of dispersion around the mean. The CS confirms the skewed nature and the CK value indicates the high kurtosis feature in the data.

\begin{table}[H]
\caption{Summary statistics for the fatigue data set.}
\centering
\begin{tabular}{lccccccccc}
\hline
Variables        & $n$  & Minimum & Median & Mean   & Maximum & SD     & CV     & CS    & CK     \\ \hline
$T_1$ & 15 & 0.243   & 1.130  & 1.247  & 2.430   & 0.700  & 56.172 & 0.209 & -1.466 \\
$T_2$         & 15 & 6.420   & 19.900   & 23.761 & 74.800  & 17.100 & 71.967 & 1.631 & 2.495  \\ \hline
\end{tabular}
\label{table:5}
\end{table}

Table \ref{table:6} presents the ML estimates and the standard errors (in parentheses) for the bivariate log-symmetric model parameters. This table also presents the log-likelihood value, and the values of the Akaike (AIC) and Bayesian (BIC) information criteria. The extra parameters were estimated using the profile log-likelihood. From Table \ref{table:6}, we observe that the log-Laplace model provides better fit than other models based on the values of log-likelihood, AIC and BIC. Nevertheless, in general, the values of log-likelihood, AIC and BIC of all bivariate log-symmetric models are quite close to each other.

\begin{table}[H]
\caption{ML estimates (with standard errors in parentheses), log-likelihood, AIC and BIC values for the indicated bivariate log-symmetric models.}
\resizebox{\linewidth}{!}{
\begin{tabular}{lccccccccc}
\noalign{\hrule height 1.7pt}
Distribuiton &          $\widehat{\eta}_1$  & $\widehat{\eta}_2$ & $\widehat{\sigma}_1$ & $\widehat{\sigma}_2$ & $\widehat{\rho}$ & $\widehat{\nu}$ & Log-likelihood & AIC &   BIC  \\ \hline
Log-normal              & 1.0362*  & 19.4824*  & 0.6536*  & 0.6210*   & -0.9390*  & -  &   -58.117  & 126.23 &129.78\\
                        & (0.0175)  & (3.1239) &  (0.01192)      & (0.1133)    & (0.0305)       &           & \\
Log-Student-$t$         & 1.0188*          & 20.1932*          & 0.6111*          & 0.5508*          & -0.9514*          & 7                                         & -57.915          & 125.83          & 129.37          \\
                        & (0.1339)  & (2.0685)   & (0.2220)       & (0.1881)  & (0.0207)       &        &           & \\
Log-Pearson Type VII    & 1.0211*          & 20.1218*          & 0.3712*          & 0.3362*          & -0.9502*          & $\xi$ = 5 , $\theta$ = 22  & -57.917          & 125.83          & 129.37          \\
                        &  (0.1806)     &  (3.2095)  & (0.0761)   & (0.0698)  & (0.0280)        &        &           & \\
Log-hyperbolic          & 1.0175*          & 20.1900*          & 0.6843*          & 0.6201*          & -0.9504*          & 2                     & -57.922          & 125.84          & 129.38          \\
                         & (0.1910)    & (3.2818)    & (0.0376)  & (0.0333)    & (0.0276)       &        &           & \\
{Log-Laplace}    & {1.0594*} & {20.9110*} & {0.7748*} & {0.6809*} & {-0.9471*} & {-}                                & {-57.585} & {125.17} & {128.71} \\
                        & (0.0023)    & (0.0105)    & (0.2032)  & (0.1745)    & (0.0342)       &        &           & \\
Log-slash               & 1.0207*          & 20.1854*          & 0.5158*          & 0.4648*          & -0.9515*          & 5                                         & -57.945          & 125.89          & 129.43          \\
                        & (0.1783)    & (3.1973)    & (0.1030)  & (0.0955)    & (0.0277)       &        &           & \\
Log-power-exponential   & 1.0298*          & 19.9461*          & 0.4516*          & 0.4154*          & -0.9432*          & 0.37                                      & -57.984          & 125.97          & 129.51          \\
                        & (0.18445)    & (3.2182)    & (0.0935)  & (0.0852)    & (0.0294)       &        &           & \\
Log-logistic            & 1.0498*          & 18.8904*          & 0.7651*          & 0.7488*          & -0.9315*          & -                                         & -58.672          & 127.34          & 130.89          \\
                        & (0.1650)    & (3.0396)    & (0.1212)  & (0.1231)    & (0.0316)       &        &           & \\\hline
\end{tabular}
}
\label{table:6}
\footnotesize{$^*$ significant at 5\% level.}
\end{table}

Figure~\ref{fig:qqplots} shows the QQ plots of the Mahalanobis distances (see Subsection \ref{maha:sec}) for the models considered in Table \ref{table:6}. We see clearly that, with the exception of the log-Student-$t$ case, the Mahalanobis distances in the bivariate log-symmetric models conform relatively well with their reference distributions.

\begin{figure}[!ht]
\centering
\subfigure[Log-normal]{\includegraphics[height=5cm,width=5cm]{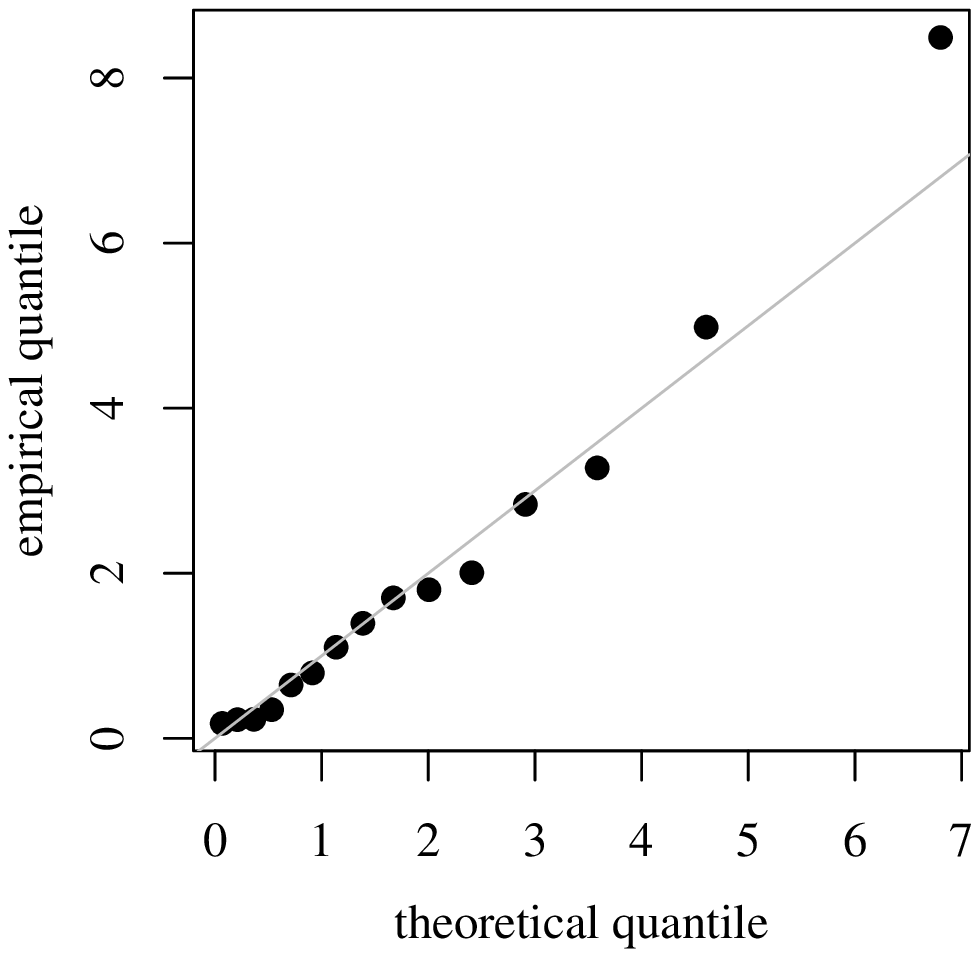}}
\subfigure[Log-Student-$t$]{\includegraphics[height=5cm,width=5cm]{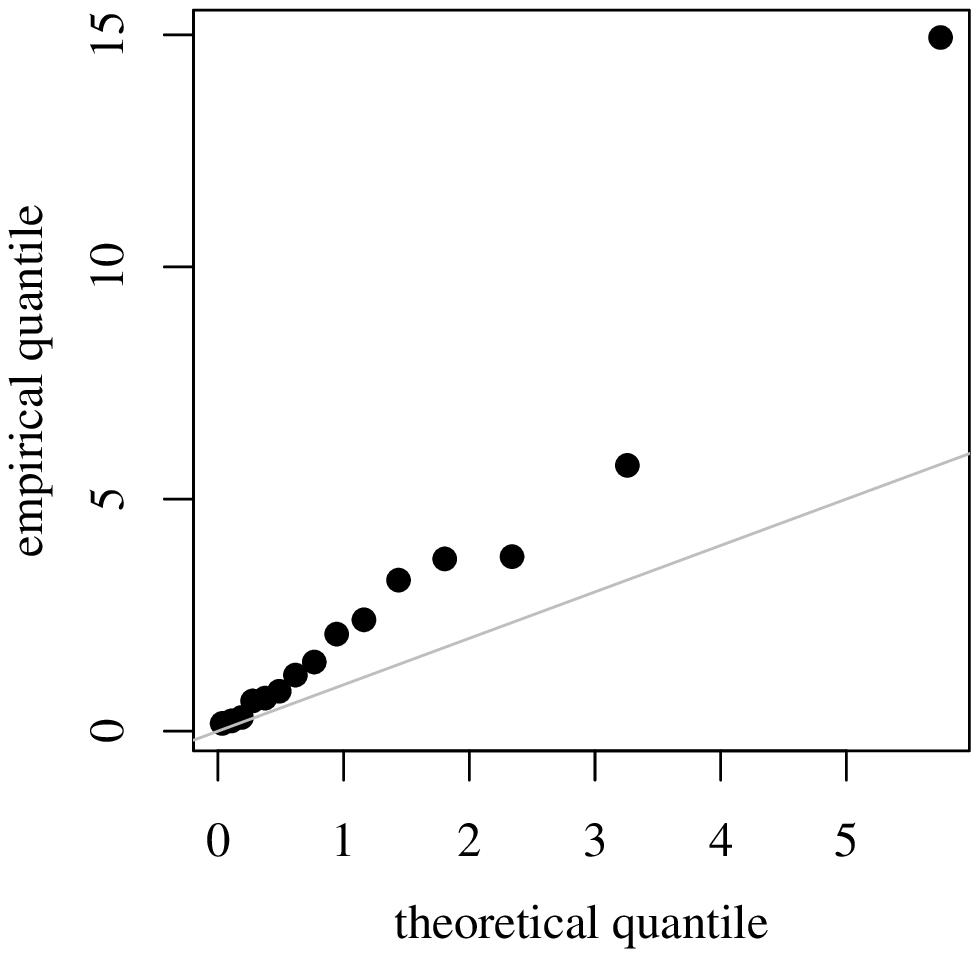}}
\subfigure[Log-Pearson Type VII]{\includegraphics[height=5cm,width=5cm]{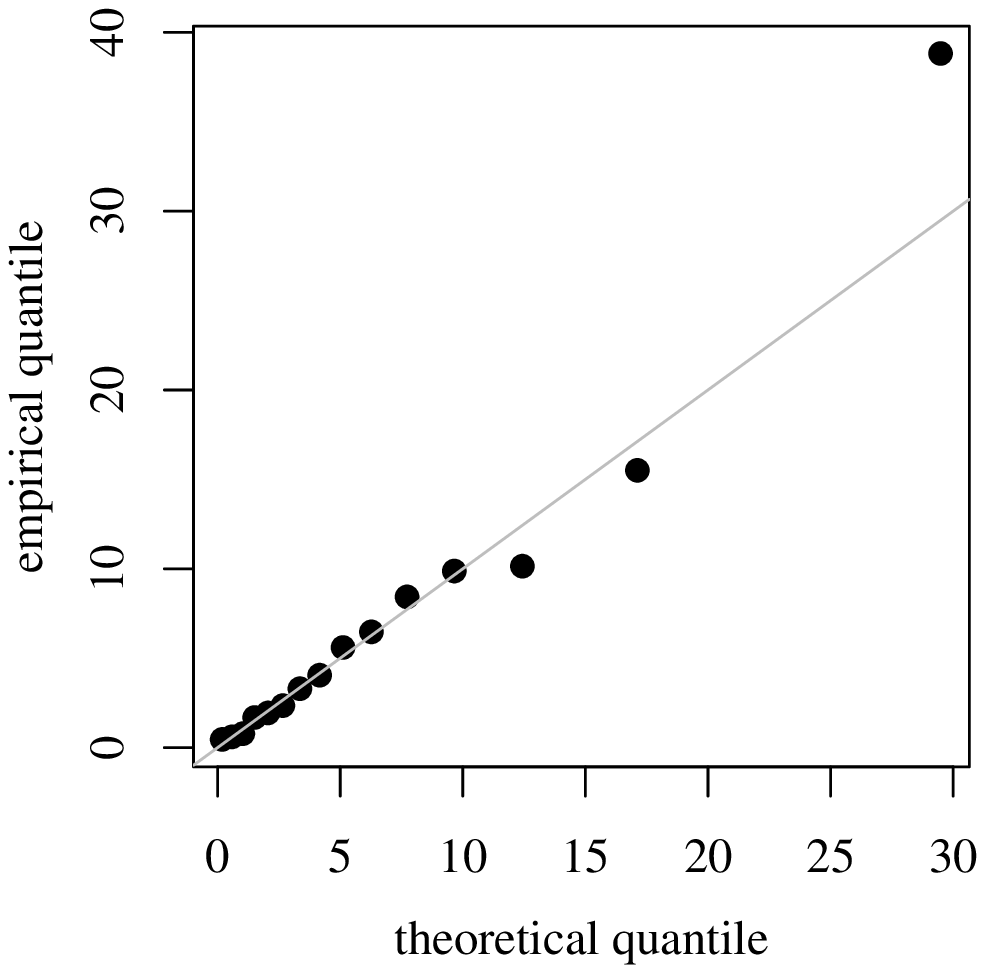}}\\
\subfigure[Log-hyperbolic]{\includegraphics[height=5cm,width=5cm]{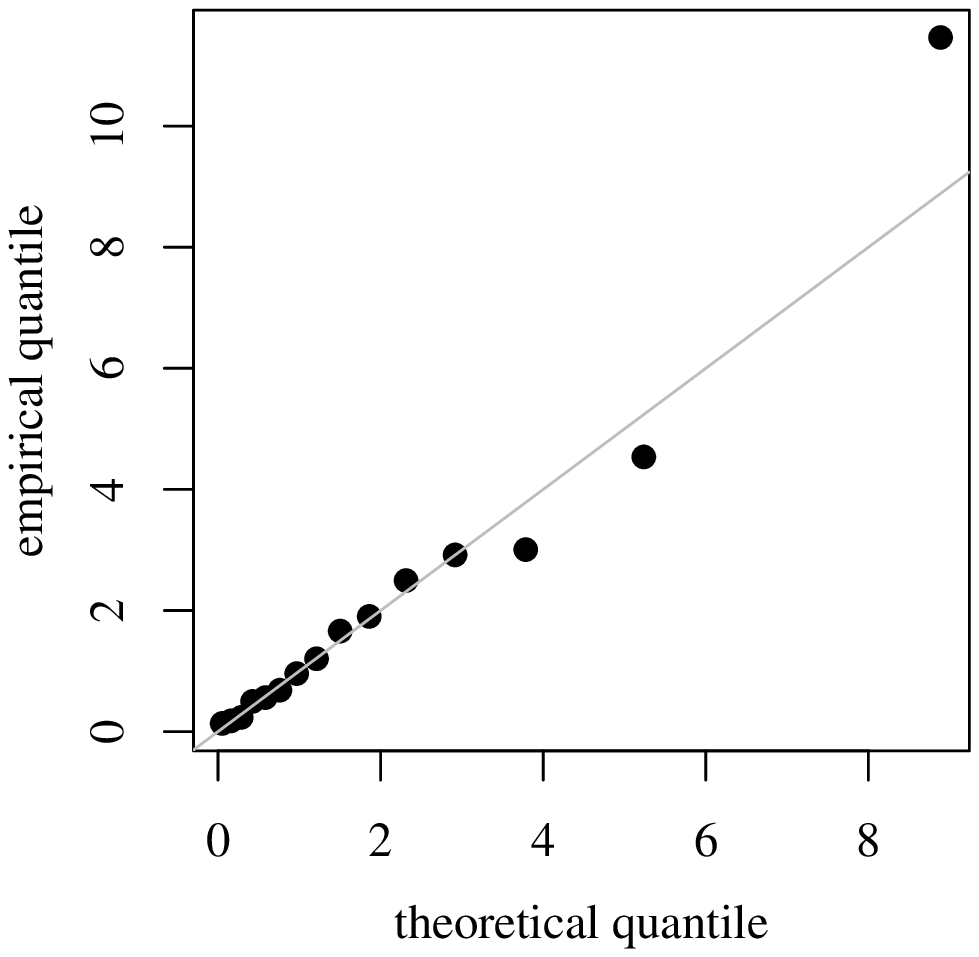}}
\subfigure[Log-Laplace]{\includegraphics[height=5cm,width=5cm]{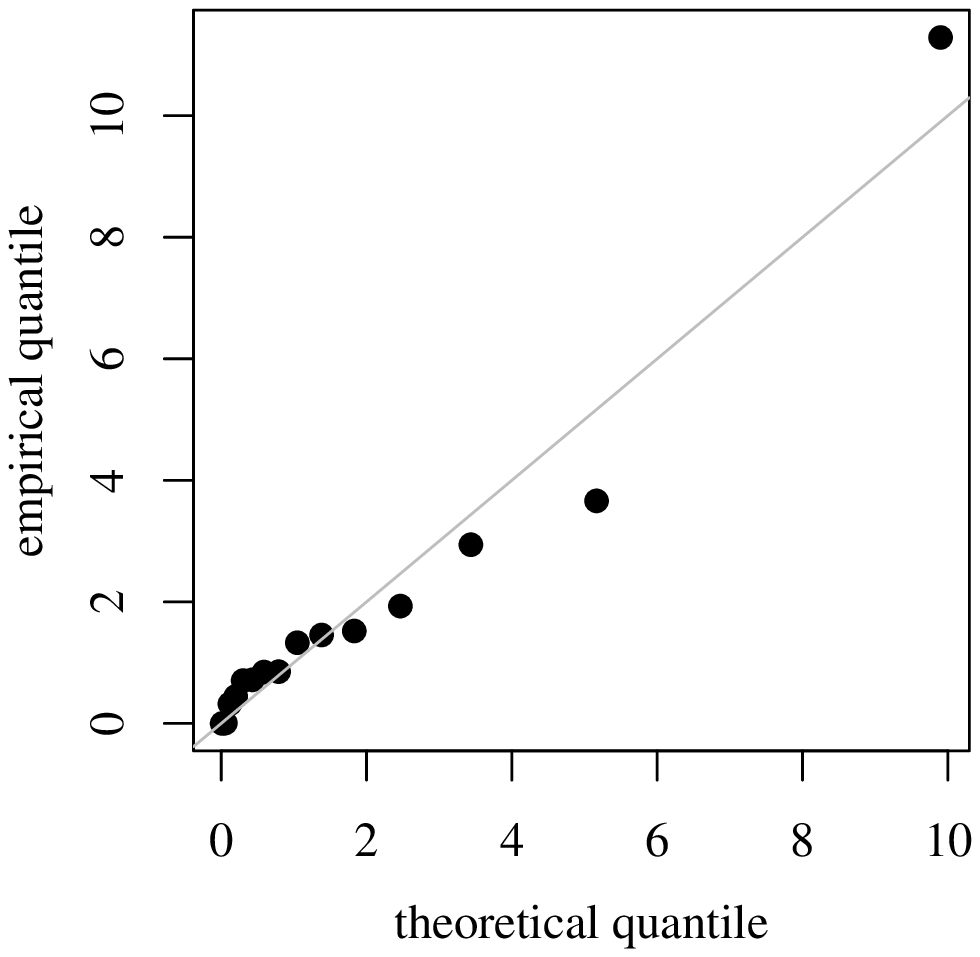}}
\subfigure[Log-slash]{\includegraphics[height=5cm,width=5cm]{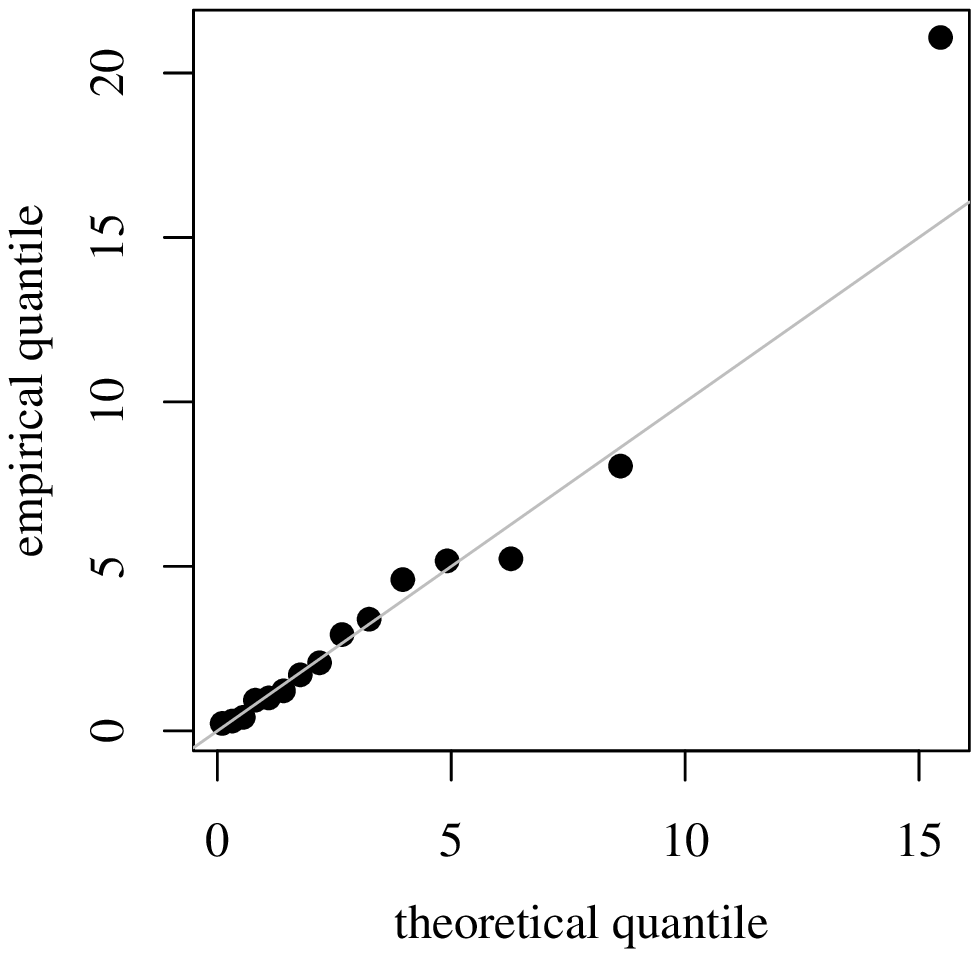}}
\subfigure[Log-power-exponential]{\includegraphics[height=5cm,width=5cm]{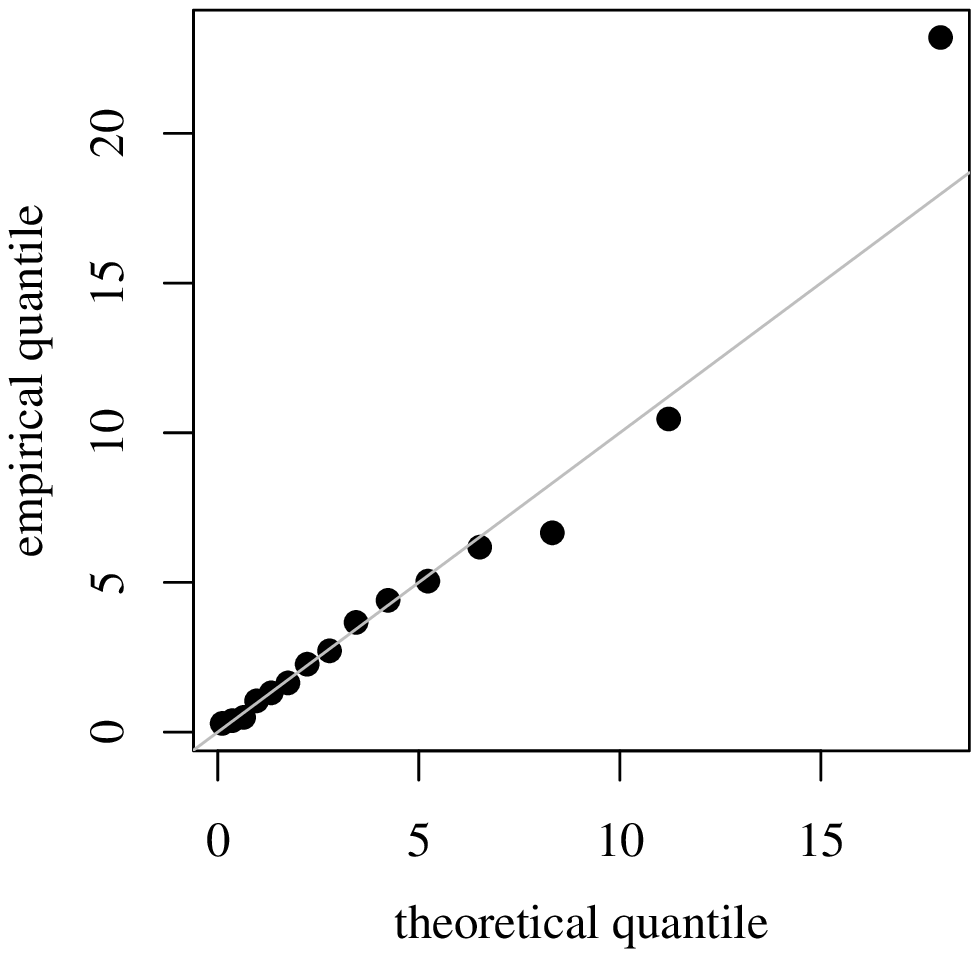}}
\subfigure[Log-logistic]{\includegraphics[height=5cm,width=5cm]{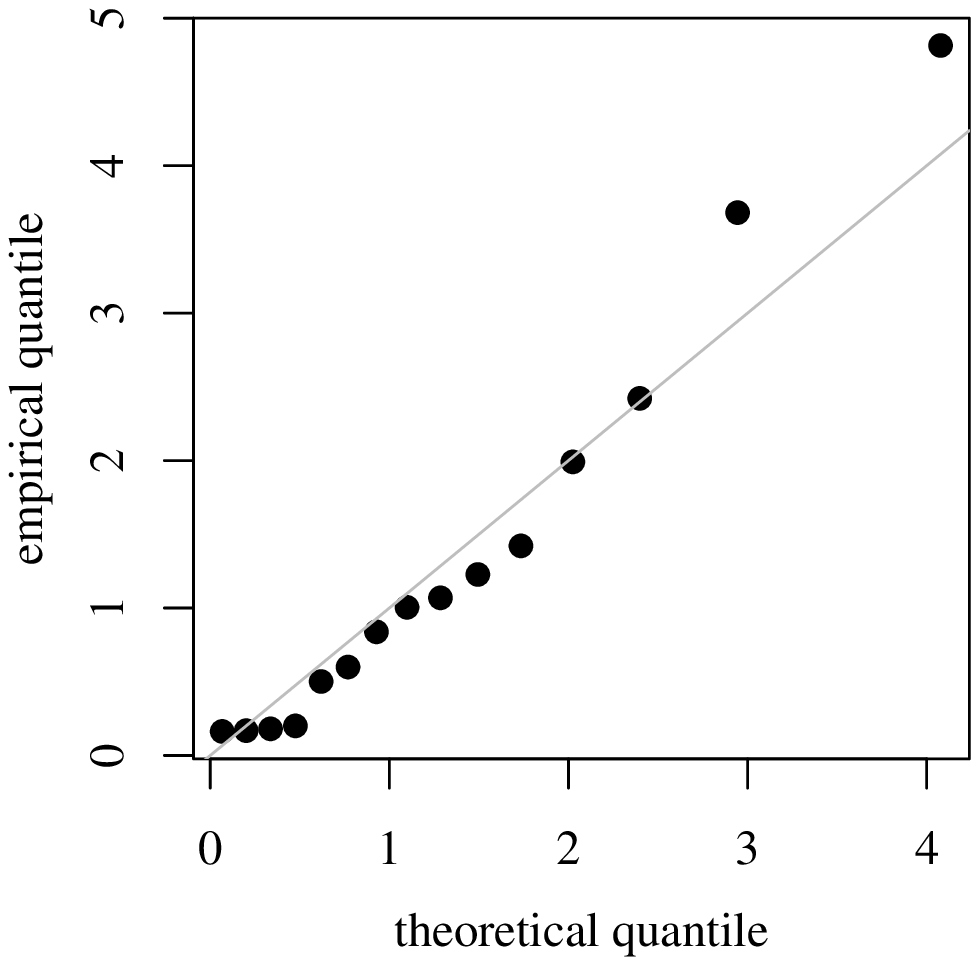}}
 \caption{\small {QQ plot of the Mahalanobis distances for the indicated models.}}
\label{fig:qqplots}
\end{figure}

\section{Concluding Remarks}\label{Sec:7}

In this paper, we have introduced a class of bivariate log-symmetric models, which is the result of an exponential transformation on a variable that follows a bivariate symmetric distribution. We have studied some of its statistical properties and also discussed the maximum likelihood estimation of the model parameters. A Monte Carlo simulation study has been carried out to evaluate the performance of the maximum likelihood estimators. The simulation results have shown that the estimators perform very well, with empirical bias values being close to zero. We have applied the proposed models to a real fatigue data set, and the results are seen to be favorable to the log-Laplace model among all the models considered. As part of future research, it will be of interest to study bivariate log-symmetric regression models. Furthermore, some hypothesis and misspecification tests (via Monte Carlo simulation) could all be studied. Work on these problems is currently in progress and we hope to report these findings in future.

\paragraph{Acknowledgements}
This study was financed in part by the Coordenação de 
Aperfeiçoamento de Pessoal de Nível Superior - Brasil (CAPES) 
(Finance Code 001). Roberto Vila and Helton Saulo gratefully acknowledge financial support from CNPq and FAP-DF, Brazil.

\paragraph{Disclosure statement}
There are no conflicts of interest to disclose.




\end{document}